\documentclass[a4paper,12pt,dvipsnames]{article}
\usepackage[T1]{fontenc}
     \usepackage[utf8]{inputenc}
     \usepackage[british]{babel}
     \usepackage{color}
     \usepackage{xcolor}

\usepackage{subcaption}
\usepackage{graphicx}
\usepackage{amsmath,amsthm,amsfonts,amssymb,mathrsfs,epsfig}
\newtheorem{theorem}{Proposition}

\theoremstyle{remark}

\usepackage{float}
\usepackage{amsmath}
\usepackage{bbm}
\usepackage{graphicx}

\usepackage{tabularx}
\usepackage{caption}
\usepackage{subcaption}
\usepackage{cite}
\usepackage{amssymb}
\usepackage{amsfonts}
\usepackage{mathtools}
\usepackage{bbm}
\usepackage{amsmath,empheq}
\usepackage{hyperref}
\usepackage[sort&compress,numbers]{natbib}
\usepackage[toc,page]{appendix}
\usepackage{amsfonts}
\usepackage{tabularx,ragged2e}
\usepackage{xcolor}
\usepackage{adjustbox}

\usepackage[outercaption]{sidecap}    
\usepackage{hyperref}

\newcommand{\mueff}{\mu_{\mathrm{eff}}}

\begin{document} 
\title{Mean first-passage time at the origin of a run-and-tumble particle with periodic forces}
\date{}

\author{Pascal Grange\\
{\emph{Division of Natural and Applied Sciences}}\\
{\emph{and Zu Chongzhi Center for Mathematics and Computational Science}}\\
 Duke Kunshan University\\
8 Duke Avenue, Kunshan, 215316 Jiangsu, China\\
\normalsize{\ttfamily{pascal.grange@dukekunshan.edu.cn}}\\
Linglong Yuan\\
{\emph{Department of Mathematical Sciences}}\\
University of Liverpool\\
Liverpool, United Kingdom\\
\normalsize{\ttfamily{linglong.yuan@liverpool.ac.uk}}}
\maketitle
\begin{abstract}
We consider a run-and-tumble particle on a half-line with an absorbing target at the origin. The particle has an internal velocity state that switches between two opposite values at Poisson-distributed times. The position of the particle evolves according to an overdamped Langevin dynamics with a spatially-periodic force field such that every point in a given period interval is accessible to the particle. The survival probability of the particle satisfies a backward Fokker--Planck equation, whose Laplace transform yields systems of equations for the moments of the first-passage time of the particle at the origin. The mean first-passage time has already been calculated assuming that the particle exits the system almost surely. We calculate the probability that the particle reaches the origin in a finite time, given its initial position and velocity. We obtain an integral condition on the force, under which the particle has a non-zero survival probability. {\textcolor{black}{The conditional average of the first-passage time at the origin (over the trajectories that reach the origin) is obtained in closed form}}. As an application, we consider a piecewise-constant force field that alternates periodically between two opposite values. In the limit where the period is short compared to the mean free path of the particle, the mean first-return time to the origin coincides with the value obtained in the case of an effective constant drift, which we calculate explicitly.  
\end{abstract}

\pagebreak

\tableofcontents

\section{Introduction, model and summary of results}

The mean first-passage time of a particle at a target is of interest in a variety of contexts, including search behavior or execution of limit orders in finance.  It has been extensively studied in the case of Brownian particles (see, for example, \cite{patie2004some, kyprianou2006introductory, duhalde2014hitting, foucart2024continuous}, and \cite{redner2001guide} for a review). The run-and-tumble particle (or RTP) {\textcolor{black}{has been introduced more recently as a model of active
particles describing the motion of bacteria}}
\cite{berg1972chemotaxis,berg2008coli,ramaswamy2010mechanics,cates2015motility}. In this model, a particle runs in a certain direction for a certain amount of time, followed by a tumble event at which the particle takes a random velocity. In dimension one, the motion of an RTP is described by the following overdamped Langevin equation: 
\begin{equation}\label{eqn:basic+F}\frac{dx_t}{dt}=F(x_t)+v\sigma(t),\end{equation}
where $x_t$ is the position of the particle at time $t$, $F:\mathbb{R}\rightarrow \mathbb{R}$ is the external force,  $v>0$ is a fixed velocity,  and $\sigma(t)$ a stochastic process switching between the values $+1$ and $-1$ with {\emph{tumbling rate}} 
$\gamma>0$ (for a formulation of the model in terms of stochastic partial differential equations, see \cite{bressloff2024stochastic}). The tumbling rate is a constant, for a model with a space-dependent tumbling rate, see \cite{singh2020run}.
We will refer to $v\sigma(t)$ as the {\emph{internal velocity state}} of the particle at time $t$. The total velocity of the particle at time $t$ is the sum of the internal velocity state and the value $F(x_t)$ of the force at the position of the particle. The propagator of the run-and-tumble particle on the real line has been calculated exactly in the free case, where the force $F$ is identically zero \cite{othmer1988models,martens2012probability,weiss2002some} (for developments in higher dimension, see \cite{mori2020universal,Santra2020position,smith2022exact}).\\

{\textcolor{black}{Diffusive particles are 
 well studied on an interval or in a semi-infinite interval (see Chapters 2 and 3 of  \cite{redner2001guide} for a review). On the other hand, exact results on the first-passage properties of RTPs are relatively sparse, even in the absence of interactions. Relaxation properties of an RTP, with and without coupling to diffusion, have been studied in \cite{malakar2018steady}, leading to the survival probability of an RTP on a half-line with an absorbing boundary. Moreover, the steady state on an interval has been shown to exhibit peaks close to the boundaries, which is in contrast with diffusive particles. Another way to obtain a steady state is to subject the RTP to stochastic resetting of its position and possibly of its velocity state \cite{evans2018run}(see \cite{Santra2020} for results in dimension two). The survival probability of a free RTP on a half-line and the non-crossing probability of two free RTPs have been calculated in \cite{le2019noncrossing}. For a study of the bound state of two interacting RTPs, see \cite{le2021stationary}, and for the statistics of the extrema  of the position, see \cite{singh2022extremal}. 
 Models of RTPs in dimension one in the presence of absorbing, hard or attractive boundaries have been studied \cite{angelani2015run,angelani2017confined,angelani2023one,bressloff2022encounter,bressloff2023encounter}, and duality properties between 
 hard and absorbing boundaries in the presence of force fields have been identified \cite{gueneau2024relating}.
  The survival probability of an RTP in confining potentials satisfies a backward Fokker--Planck equation, which has been solved exactly in Laplace space in the harmonic case in \cite{dhar2019run}. The survival probability and propagator in a linear confining potential have been studied in \cite{nath2024survival}. The mean first-passage time at the origin in confining potentials has been optimized in the tumbling rate in \cite{gueneau2024optimal}.}}\\


In this paper, we restrict ourselves to the one-dimensional case where there is only one absorbing boundary at the origin, and where the force $F$ is periodic. Given the initial position of an RTP on the positive half-line, we calculate the exit probability (the probability of eventually reaching the origin, which is the complementary event to survival), and the mean exit time conditional on exit. If there are two absorbing boundaries $a, b$ with $a<b$, and the RTP starts within $[a,b]$, the hitting probabilities for $a$ and for $b$ have been explicitly computed in \cite{gueneau2024relating}  for any $F$. We use this result and a Markov-chain argument to express the exit probability in our model. The mean exit time has been computed explicitly in \cite{gueneau2024run} for some special forms $F$ that ensure almost-sure exit for the particle.  For a model with a periodic piecewise-linear potential, see \cite{roberts2023run}.  
 Run-and-tumble particles with periodic force fields were studied in \cite{le2020velocity}, with emphasis on the distribution of the position (the mean first-passage time at a given level was obtained assuming reflecting boundary conditions at the end of a half-line). {\textcolor{black}{The present one-dimensional system is highly idealized, but one may think of it as a model 
 of a situation in which self-propelled particles (bacteria for instance) would be produced in an environment with a succession of traps. The traps could correspond to gradients of concentrations of signaling molecules. The traps form a repetitive pattern along the direction of the half-line in the model (and their extension in the transverse direction is neglected).}}\\  

More specifically, we consider a run-and-tumble particle in one dimension, described by Eq. \eqref{eqn:basic+F},  with velocity $v$ and tumbling rate $\gamma$.  The positive half-line is endowed with a periodic force field $F$ (with period denoted by $a$): for any $x\geq 0$,
\begin{equation}
 F( x+a) = F(x).
\end{equation}
 We want the particle to be able to reach any point in the positive half-line. The force field is therefore assumed to satisfy 
\begin{equation}\label{eqn:<v}
 |F(x)|< v
\end{equation}
 for all $x\geq 0$ (the values of the force field are said to be subcritical). 
However, the above condition is not sufficient for all points within a period interval (including endpoints) to be accessible to a particle starting at any point in this interval. We need to make the following two assumptions:\\   
$(i)$ the equation 
\begin{equation}\label{eqn:t+}
\frac{dx_t}{dt}=F(x_t)+v,\quad  x_0=0,
\end{equation}
induces a finite time $t_+>0$ such that $x_{t_+}=a$;\\
$(ii)$ the equation  \begin{equation}\label{eqn:t-}
\frac{dx_t}{dt}=F(x_t)-v,\quad  x_0=0, \end{equation}
induces a finite time $t_{-}>0$ such that $x_{t_-}=-a$.\\
 Clearly, these two conditions imply \eqref{eqn:<v}. To our best knowledge, conditions $(i)$ and $(ii)$ have not been mentioned in the literature, but are necessary and sufficient to guarantee accessibility of all points within a period interval (including endpoints). Let us give a counterexample for condition $(i)$. Let us define $F$ on a period interval by $F(x)=-v+a^{-1}\gamma (a-x)^2$ for $x\in[0,a[$, assuming $a< 2v\gamma^{-1}$ (so that the values of $F$ are subcritical). If the particle starts at the origin,  
 Eq. 
\eqref{eqn:t+} becomes
\begin{equation}
\frac{dx_t}{dt}= a^{-1}\gamma (a-x_t)^2,\qquad  x_0=0.
\end{equation}
 Solving this differential equation yields $x_t = a \gamma t(1 + \gamma t)^{-1}$. 
 Hence $\lim_{t\to\infty}x_t=a$, but the particle will never reach $a$ and go beyond.\\

{\textcolor{black}{
 On the other hand, let us consider the following modification of the above force field:
\begin{equation}\label{FEpsilonDef}
 F_\epsilon( x ) = -v( 1 -\epsilon ) + a^{-1}\gamma( a-x)^2, \qquad(x\in [0,a[ ).
\end{equation}
 with $\epsilon$ satisfying the condition
 \begin{equation}\label{epsilonCond}
 0 < \epsilon < 2- \frac{a\gamma}{v}.
 \end{equation}
 The upper bound in the above conditions is positive because we still assume $a<2v\gamma^{-1}$. In this case, $F_\epsilon$ takes only subcritical values, i.e. \eqref{eqn:<v} is satisfied.}}\\

{\textcolor{black}{
  Solving Eq. (\ref{eqn:t+}) with $F=F_\epsilon$, we obtain
\begin{equation}
 x_t = a - \sqrt{\frac{av\epsilon}{\gamma}}\tan\left( \operatorname{arctan}\left(\sqrt{\frac{\gamma a }{v\epsilon}} \right) - \sqrt{\frac{\gamma v\epsilon}{a}} t \right),
\end{equation}
 which holds as long as $x_t$ is in the interval $[0,a]$. The  particle reaches $a$ at $t_+$ which we read off as
\begin{equation}
 t_+ = \sqrt{\frac{a}{\gamma v\epsilon}}\operatorname{arctan}\left(\sqrt{\frac{\gamma a }{v\epsilon}} \right).
 \end{equation}
 This time goes to infinity when $\epsilon$ goes to $0$, which is intuitively consistent with the previous counterexample (which is the limit of $F_\epsilon$ when $\epsilon$ goes to zero).}}\\

{\textcolor{black}{
Solving Eq. (\ref{eqn:t-}) with $F=F_\epsilon$, we obtain
\begin{equation}
x_t = -\sqrt{\frac{va}{\gamma}(2-\epsilon)}
\frac{e^{\sqrt{\frac{v\gamma(2-\epsilon)}{a}}t} - e^{-\sqrt{\frac{v\gamma(2-\epsilon)}{a}}}t}{
e^{\sqrt{\frac{v\gamma(2-\epsilon)}{a}}t} + e^{-\sqrt{\frac{v\gamma(2-\epsilon)}{a}}t}
}.
\end{equation}    
 which holds as long as $x_t$ is in the interval $[-a,0]$. The assumption on $\epsilon$ we made in Eq. (\ref{epsilonCond}) ensures that $2-\epsilon>0$ and $\frac{v\gamma(2-\epsilon)}{a}>1$. The expression $x_t$ therefore goes to a value lower than $-a$ at large time. 
 The particle reaches the position $-a$ at
\begin{equation}
 t_- = \sqrt{\frac{a}{v\gamma( 2 - \epsilon)}}\operatorname{atanh}\left( \sqrt{\frac{a\gamma}{v(2-\epsilon)}}\right).
\end{equation}
}}
{\textcolor{black}{ The two situations are depicted in Fig. \ref{figCounterExamples}.}}

\begin{figure}[H]
\begin{subfigure}{\textwidth}
    \centering
    \includegraphics[width=0.98\textwidth]{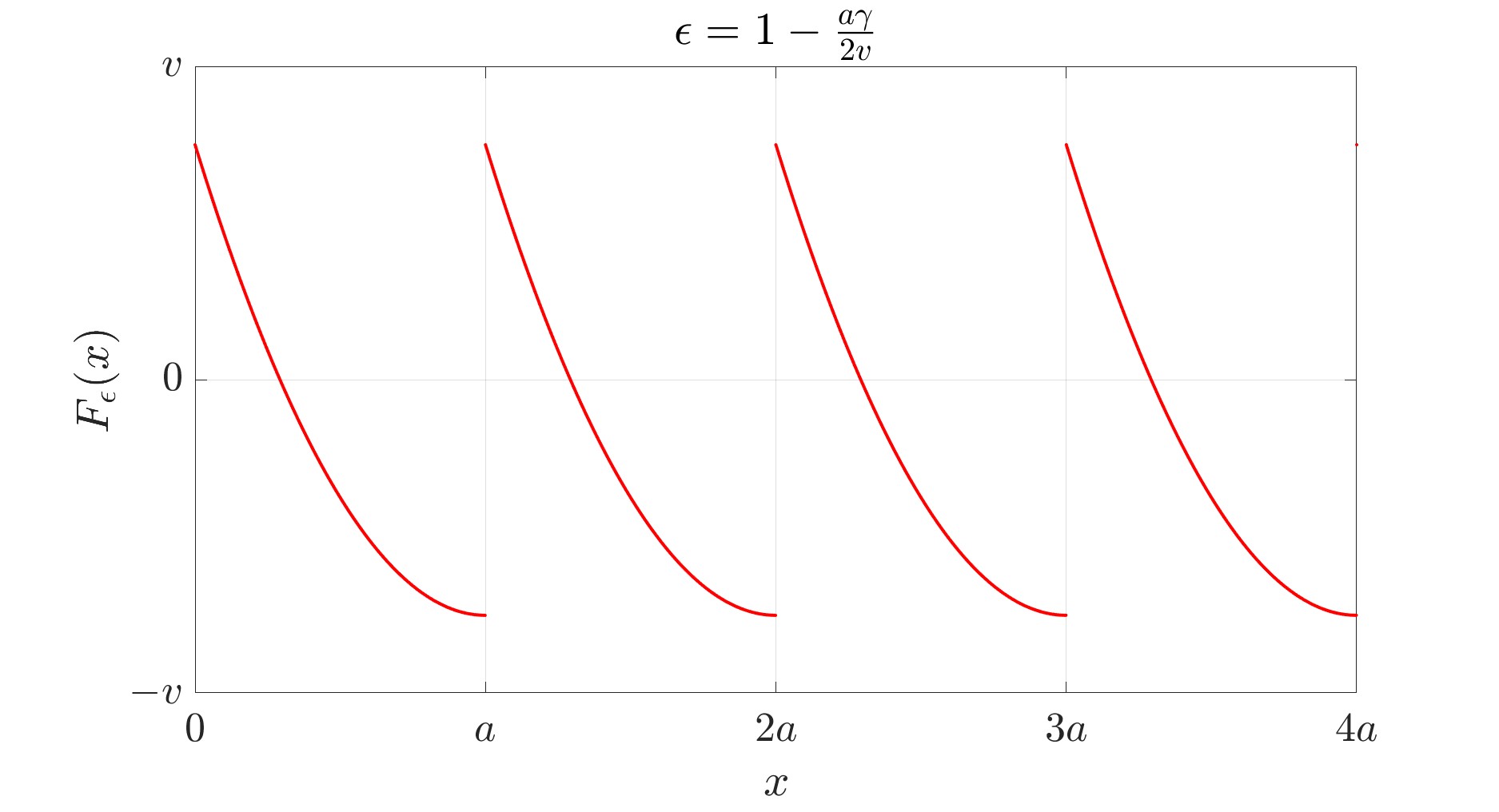}
    \caption{}
    \label{fig:second}
\end{subfigure}
\begin{subfigure}{\textwidth}
    \centering
    \includegraphics[width=0.98\textwidth]{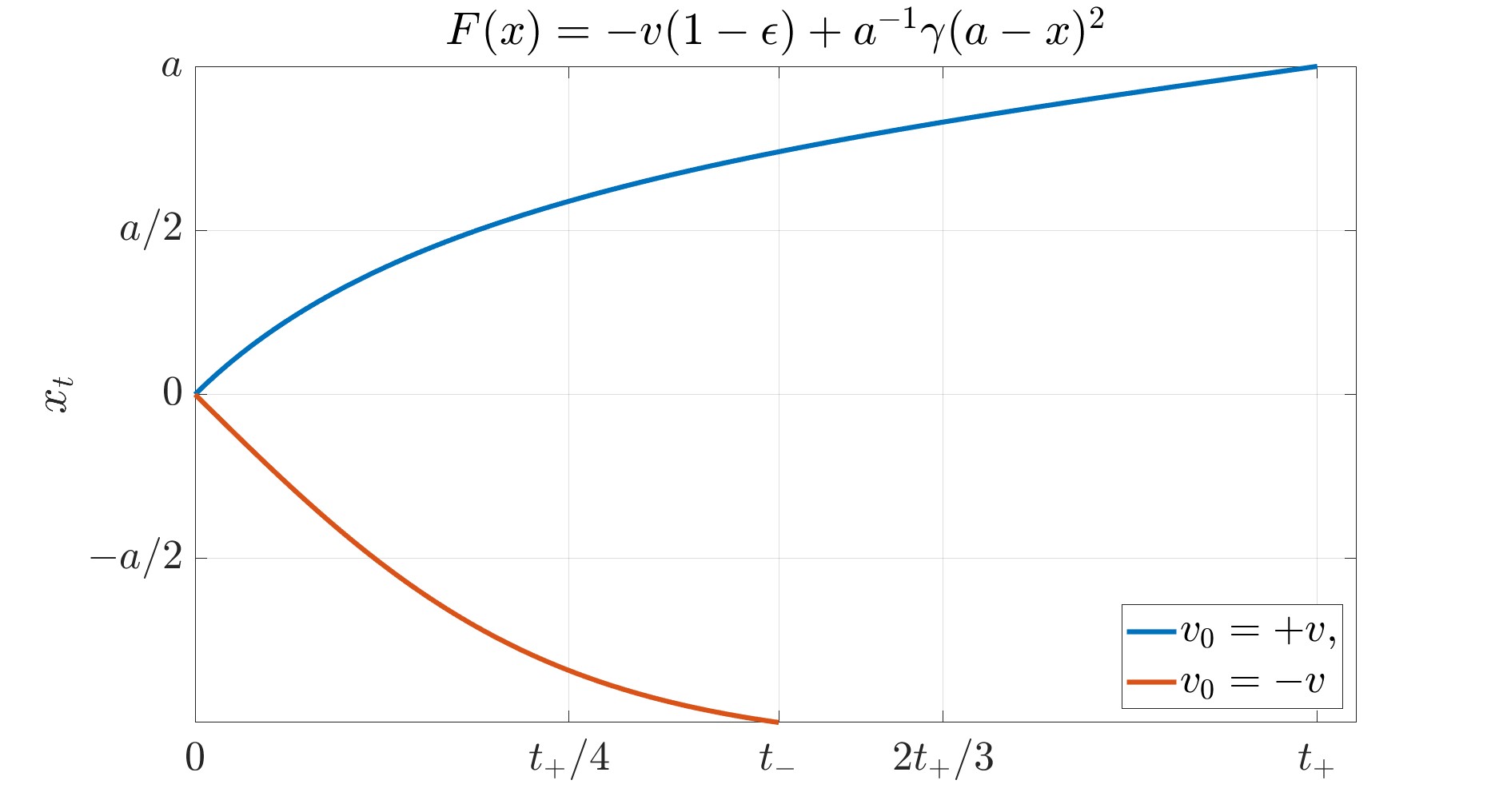}
    \caption{}
    \label{fig:third}
\end{subfigure}
\caption{ (a) The periodic potential $F_\epsilon$ defined in Eq. (\ref{FEpsilonDef}), for which all points in the positive half-line are accessible.
(b) The solutions to the differential equations displayed in Eqs (\ref{eqn:t+},\ref{eqn:t-}), on the intervals of time for which the particle is is within distance $a$ from the origin. Numerically the parameters were set to $\gamma = 1$, $v = 1$, $a=3/2$ and  $\epsilon = (2 - a \gamma / v)/2$ so that the condition given in Eq. (\ref{epsilonCond}) is satisfied.}
\label{figCounterExamples}
\end{figure}


{\textcolor{black}{Assuming all points in a period interval including the endpoints are accessible, we first study the exit probability, given an initial position $x\geq 0$ and an initial internal velocity state:
\begin{equation}
\begin{split}
&E(x,\pm):= P( {\mathrm{min}}( t\geq 0, x_t = 0) < \infty | x_0 = x>0, \sigma( 0 ) = \pm 1),\\
&E(0,+):= P( {\mathrm{min}}( t>0, x_t = 0) < \infty | x_0 = 0, \sigma( 0 ) = 1),\\
&E(0,-):= P( {\mathrm{min}}( t\geq 0, x_t = 0) < \infty | x_0 = 0, \sigma( 0 ) = -1)=1.
\end{split}
\end{equation}}}
 For convenience, let us work in a system of units where the unit of length is the mean free path $v\gamma^{-1}$ and the
unit of time is the mean $\gamma^{-1}$ of the Poisson-distributed times between tumble events. This is achieved\footnote{{\textcolor{black}{For an explicit example of a change of variables mapping ordinary space coordinates to coordinates in our system of units (where the unit of length in $v\gamma^{-1}$), see Appendix \ref{appUnits}.}}} by setting the parameters $v$ and $\gamma$ to $1$. The exit probability satisfies the following system of equations:
\begin{equation}\label{sysE}
\begin{cases}
  0 =& (F(x) + 1) \frac{\partial  E(x,+)}{\partial x}
   +   E(x,- ) -   E(x,+ ),\\
  0 =& (F(x) - 1) \frac{ \partial  E(x,- )}{\partial x}
   -  E(x,- ) +   E(x,+ ).
\end{cases} 
\end{equation}

 Another quantity of interest is the mean exit time $T(x,\pm)$ of the particle that starts at coordinate $x$ with internal velocity state $\pm 1$, calculated over the configurations with finite exit times:
\begin{equation}
 T( x, \pm ) =  \int_0^\infty dt \, t \Phi_\pm(x,t),      
\end{equation}
 where $\Phi_\pm(x,t)$ is the flow density at time $t$ through the origin of a particle that was at coordinate $x$ with internal velocity state $\pm$ at time $0$. 
     The quantities $T(x,\pm)$  satisfy the system of equations
\begin{equation}\label{sysT}
\begin{cases}
 -E(x,+) = (F(x) + 1) \frac{\partial T(x,+)}{\partial x}
   + T(x,-) -  T(x,+),\\
 -E(x,-) = (F(x) - 1) \frac{\partial 
T(x,-  )}{\partial x}
   -T(x,-)  +  T(x,+).\\
\end{cases}
\end{equation}
 The conditional average of the first passage-time at the origin (over the trajectories that eventually reach the origin) can be defined in all cases as 
 \begin{equation}\label{condDef}
     \langle T(x)\rangle_{c,\pm} := \frac{T(x,\pm)}{E(x,\pm)},
\end{equation}
 where $c$ in the above notation (set in \cite{de2021survival}) indicates that the average is conditional on the exit of the particle\footnote{$\langle T(x)\rangle_{c,\pm}$ will sometimes be referred to as {\emph{the conditional average of the first-passage time}} for short.}. In cases where the particle with a given initial velocity state exits the system almost surely,  
  the conditional average $\langle T(x)\rangle_{c,\pm}$ is the mean first-passage time at the origin, which coincides with $T(x,\pm)$.\\

The exit probabilities $E(x,\pm)$ and the conditional averages $\langle T(x)\rangle_{c,\pm}$ were calculated in \cite{de2021survival} in the case where $F$ consists of a constant drift. For subcritical values of the drift, two regimes were identified, depending on the sign of the drift. The particle exits the system almost surely if the drift is negative, otherwise its exit probability exit decreases exponentially with the starting distance from the origin. More precisely,  the exit probabilities and conditional averages of the first-passage time at the origin corresponding statements are as follows.\\
 
{\bf{$(\mathrm{I})_{\mu}$ Exit probability with a constant drift  $\mu$.}} If $F(x)=\mu$ for all $x\geq 0$, then
\begin{equation}\label{EConst}
\begin{split}
 E(x,+) =& \mathbbm{1}( \mu \leq 0 ) + \frac{1-\mu}{1+\mu}e^{-\frac{2\mu}{1-\mu^2} x}\mathbbm{1}( \mu > 0 ),\\
E(x,-) = & \mathbbm{1}( \mu \leq 0 ) + e^{-\frac{2\mu}{1-\mu^2} x}\mathbbm{1}( \mu > 0 ).\\ 
\end{split} 
\end{equation}

  The mean first-passage time at the origin is infinite if the drift is equal to zero. Otherwise,  the conditional averages $\langle T(x,\pm)\rangle_c$ are as follows.\\
  
{\bf{$(\mathrm{II})_{\mu}$ Conditional average of the first-passage time over trajectories that reach the origin, with a constant drift $\mu$.}}  If $F(x)=\mu\neq 0$ for all $x\geq 0$, then $\langle T(x)\rangle_{c,+}$  is finite if and only if $\mu\neq 0$. Moreover, 
\begin{equation}\label{resIImu}
\begin{split}
 \langle T(x)\rangle_{c,+}=&  \left( \frac{x}{|\mu|} + \frac{1}{|\mu|}   \right)\mathbbm{1}( \mu < 0 ) + \left(\frac{1+\mu^2}{\mu(1-\mu^2)}x + \frac{1}{\mu}  \right) \mathbbm{1}( \mu > 0 ),\\
 \langle T(x)\rangle_{c,-} = &  \frac{x}{|\mu|}\mathbbm{1}( \mu < 0 ) +   \frac{1+\mu^2}{\mu(1-\mu^2)}x \mathbbm{1}( \mu > 0 ).\\ 
\end{split} 
\end{equation}


 The problem of Eq. (\ref{sysT}) was solved in \cite{gueneau2024relating} for a general force field $F$, assuming the exit probabilities $E(x,\pm)$ are identically equal to $1$.\\ 

In this paper, we prove the following generalization of $(\mathrm{I})_{\mu}$ for a periodic force field (with subcritical values such that all the points in a period interval are accessible to the particle):\\

$(\mathrm{I})_{F}$ {\bf{Exit probability with a periodic force field.}} If $F$ is a periodic force field of period $a$, the particle exits the system almost surely if  
\begin{equation}\label{condExit}
  J(a)\leq 0,\quad {\mathrm{with}} \quad J(x):= \int_0^x \frac{2F( y ) }{ 1 - F(y)^2} dy.
\end{equation}
  Moreover, the exit probability given the initial velocity state is expressed as
\begin{equation}\label{reportedE}
\begin{split}
E( x+Na,\pm) =& \mathbbm{1}\left( J(a) \leq 0 \right)\\
 &+ \left[  \frac{
  e^{-NJ(a)}[2\Xi_-(a) - (1-e^{-J(a)})(2\Xi_-(x)\pm e^{-J(x)}) ] }{1-e^{-J(a)} + 2\Xi_-(a)}\right] \mathbbm{1}\left( J(a) > 0 \right),\\
  &{\mathrm{for}}\quad x\in [0,a[,\;\;\;\; N \in\mathbb{N},\quad
  {\mathrm{with}}\;\;\;\Xi_\pm(x)=  \int_0^x
     \frac{e^{\pm J(y)}}{1 - F(y)^2}    dy.
\end{split}
\end{equation}

For the conditional average of the first-passage time to the origin, we prove the following generalization of $(\mathrm{II})_{\mu}$ in the case of a periodic force field:\\

{\bf{$(\mathrm{II})_{F}$ Conditional average of the first-passage time over trajectories that reach the origin, with a periodic force field  $F$.}} 
The conditional average of the first-passage time at the origin is infinite if $J(a)=0$. Otherwise, it is given by an affine function of $N$, if the starting point is in the interval $[Na, (N+1)a[$.\\
$\bullet$ If $J(a)<0$,
\begin{equation}\label{reportedTNeg}
\begin{split}
\langle T( x+Na )\rangle_{c,\pm} =& T( x+Na )\\
=&\left(  \frac{J(a)}{2} +2\frac{ \Xi_+( a) \Xi_-(a)}{1 - e^{J(a)}} -2 \int_0^a dv\Xi'_-(v)\Xi_+(v) \right)N\\
&+ \frac{\Xi_+(a)}{1-e^{J(a)}}\left( 1 \pm e^{-J(x)} + 2\Xi_-(x)\right)
 +\frac{1}{2}J(x) \mp e^{-J(x)}\Xi_+(x) - 2 \int_0^x dv\Xi_-'(v)\Xi_+(v),\\
   &{\mathrm{for}}\qquad x\in [0,a[,\;\;\;\; N \in\mathbb{N}.
\end{split}
\end{equation}
$\bullet$ If $J(a)>0$, 
\begin{equation}\label{reported}
\begin{split}
\langle T(x+Na)\rangle_{c,\pm}
=& \left(J(a ) + 2\frac{\Xi_+(a) \Xi_-(a)}{1- e^{-J(a)}} -2 \int_0^a dv \Xi'_+(v) \Xi_-(v)\right)N\\ &+ 
\frac{(-\psi(a) - 2\frac{\Xi_-(a)}{1 - e^{-J(a)}}\varphi(a))(1 \pm e^{-J(x)} + 2\Xi_-(x)) + ( \psi(x) \pm \varphi(x) )(1- e^{-J(a)}+ 2\Xi_-(a) ) }{2\Xi_-(a) - (1-e^{-J(a)})(2\Xi_-(x) \pm e^{-J(x)})  }
,\\
  &{\mathrm{for}}\qquad x\in [0,a[,\;\;\;\; N \in\mathbb{N},
\end{split}
\end{equation}
where the functions $\varphi$ and $\psi$ are defined on $[0,a]$ by
\begin{equation}\label{phiPsiDefInt}
\begin{split}
\varphi(x)=& \frac{1}{2}e^{-J(x)} \int_0^x\frac{[E(x,+)-E(x,-)]F(y) - [E(x,+)+E(x,-)]}{1-F(y)^2}e^{J(y)} dy,  \\
\psi(x) =& 2\int_0^x \frac{\varphi(y)}{1-F(y)^2} dy + \frac{1}{2}\int_0^x\frac{[E(x,+)+E(x,-)]F(y) -[E(x,+)-E(x,-)]}{1-F(y)^2}dy, \qquad( x\geq 0).\\
\end{split}
\end{equation}



The paper is organized as follows.  In Section \ref{QOI} we review the backward Fokker--Planck equation satisfied by the survival probability of an RTP in a force field with subcritical values and satisfying \eqref{eqn:t+} and \eqref{eqn:t-}. Taking the Laplace transform induces the systems of differential equations presented in Eqs (\ref{sysE},\ref{sysT}). In Section \ref{ExitProbability}, we restrict the discussion to periodic force fields and work out the condition under which the exit time is almost surely finite. We solve Eq. (\ref{sysE}) and obtain the result reported in Eq. (\ref{reportedE}). In Section \ref{MFPT}, we solve Eq. (\ref{sysT}), working separately in the regimes $J(a)\leq 0$ and $J(a)>0$. The periodicity of the force field is used to derive one of the boundary conditions from a renewal argument. Indeed, a particle in a positive velocity state sees the same system whether it is at position $0$ or at position $a$ (the other boundary condition is more general, as a particle passing at $x=0$ in a negative velocity state leaves the system immediately). We obtain a closed-form expression of $T(x,\pm)$. We turn this expression into the form presented in Eq. (\ref{reportedTNeg},\ref{reported}), a renewal argument. The result is confirmed by direct calculations based on periodicity, presented in Appendices \ref{appInt} and \ref{periodApp}. As an application, the exit probability and the mean first-return time to the origin are evaluated explicitly (in Section \ref{piecewiseConst}) for a piecewise-constant force field alternating between two opposite values. In Appendix \ref{checkConst}, we check that the statement $(\mathrm{II})_{F}$ reproduces $(\mathrm{II})_\mu$ in the case of a constant force field.\\ 


\section{Notations and quantities of interest}\label{QOI}

 Consider a run-and-tumble particle on the positive half-line, starting its motion at position $x>0$, with a given internal velocity state $\pm 1$ (the unit of velocity is set to the absolute value of the internal velocity).
 The particle undergoes Poisson-distributed tumble events. At each tumble event, the internal velocity state flips instantaneously. The rate $\gamma$ of the corresponding Poisson process (the tumbling rate) is set to $1$ ($\gamma^{-1}$ is chosen as the unit of time). The particle is said to be alive as long as its coordinate is positive. If the particle reaches the origin {\textcolor{black}{with negative velocity}}, the process stops. The time (if any) at which the particle reaches the origin, is called the exit time.\\ 
 
 Let us assume that the positive half-line is endowed with a field $F$ satisfying the conditions $(i)$ and $(ii)$ of Eqs (\eqref{eqn:t+},\eqref{eqn:t-}), which implies \eqref{eqn:<v}. 
 Under this assumption, the total velocity of the run-and-tumble particle, which is the sum $F(x_t) + \sigma(t)$ is nonzero for the duration of the process. Moreover, the sign of the total velocity of the particle is the same as the sign of the internal velocity $\sigma(t)$ at all times.\\

 Consider the survival probability $S_\pm(x,t)$ at time $t$ of a run-and-tumble  particle that was at position $x$ in the internal velocity state $\pm 1$ at time $0$. This survival probability satisfies a backward Fokker--Planck equation \cite{bray2013persistence,dhar2019run}, which is derived as follows. For a particle that is alive at time $t$, and has been at position $x>0$ at time $0$ in the internal velocity state $\pm 1$, consider the time interval $[0, dt]$, where $dt$ is an infinitesimal time. If there is no tumble event in this interval (which is the case with probability $1- dt$), the particle is at position $x+[F(x) \pm 1]dt$ at time $dt$. It must survive for a duration of time $t-dt$ to be alive at time $t$. If there is a tumble event in the interval $[0,dt]$ (which is the case with probability $dt$), the internal velocity state flips, and the particle has to survive until time $t$. Hence the two terms in the expression of the survival probability $S_\pm( x,t)$:
 \begin{equation}
 S_\pm( x,t) =  ( 1- dt )S_\pm( x + [F(x) \pm 1]dt,t-dt) +
   dt S_\mp( x,t ).
 \end{equation}
Taylor expansion yields
\begin{equation}
 S_\pm( x,t) =  S_\pm(x,t) + [F(x) \pm 1]dt \partial_x S_\pm(x,t) - dt \partial_t S_\pm(x,t) -  dt S_\pm( x,t) +
dt S_\mp( x,t).\\ 
\end{equation}
Hence the evolution equation satisfied by the survival probability on the positive half-line:
\begin{equation}\label{Seq}
 \partial_t S_\pm(x,t) =  [F(x) \pm 1] \partial_x S_\pm(x,t) -  S_\pm( x,t) +   S_\mp( x,t),\quad\quad(x>0,t\geq 0).\\ 
\end{equation} 

Let us denote by $\Phi_\pm(x,t)$ the density of first-passage time at the origin with a negative total velocity of a particle that started its motion (at time $0$) at position $x$  with  the internal velocity state $\pm v$:
\begin{equation}
\begin{split}
 \Phi_\pm(x,t) =& -\frac{\partial S_\pm(x, t)}{\partial t}.\\
 \end{split}
\end{equation}
With this definition, $\Phi_+(0,0) = 0$ because the total velocity of a particle that is at zero at time $0$ in a positive internal velocity state is $F(0)+1>0$ (the particle does not exit the half-line at time $0$). 
 Differentiating Eq. (\ref{Seq}) w.r.t. $t$ yields a system of coupled evolution equations for the densities $\Phi_+$ and $\Phi_-$:\\
 \begin{equation}\label{Phieq}
 \begin{split}
\partial_t \Phi_\pm(x,t) =& [  F(x) \pm 1] \partial_x \Phi_\pm(x,t) -  \Phi_\pm(x,t) +  \Phi_\mp (x, t).\\
\end{split}
 \end{equation}

 {\textcolor{black}{Let us denote by $\tilde{\Phi}$ the Laplace transform of the flow through the origin:
\begin{equation}
 \tilde{\Phi}_\pm(x,s):= \int_0^\infty dt\, e^{-st} \Phi_\pm( x,t).
\end{equation}
  Taking the Laplace transform of Eq. (\ref{Phieq})}}  (assuming $e^{-st}\Phi_\pm(x,t)$ goes to zero when $t$ goes to infinity), yields the following system of equations for  $\tilde{\Phi}_\pm(x,s)$ (with $x>0$):
 \begin{equation}\label{PhiTildeEq}
 s\tilde{\Phi}_\pm(x,s) = [F(x) \pm 1] \partial_x 
  \tilde{\Phi}_\pm(x,s) -  \tilde{\Phi}_\pm(x,s) + \tilde{\Phi}_\mp(x,s). 
 \end{equation} 
We have used the fact that $\Phi_\pm(x,0)=0$ for $x>0$ (the particle does not reach the origin at time zero if it is at $x>0$).\\

 Borrowing the notations of \cite{gueneau2024relating}, let us denote by $E(x,\pm)$ the probability that a particle starting its motion at time $0$ at position $x$ with internal velocity state $\pm 1$ exits the positive half-line  in finite time
 \begin{equation}
 E(x,\pm) = \int_0^\infty dz  \Phi_\pm(x,z)  = \tilde{\Phi}_\pm(x,0).
 \end{equation}
 This exit probability appears at order $0$ in the Taylor expansion of $\tilde{\Phi}_\pm(x,s)$ around $s=0$ as
\begin{equation}\label{LapExp}
\begin{split}
\tilde{\Phi}_\pm(x,s):=& 
 \int_0^\infty dt \,t \Phi_\pm(x,t) + o(s)\\
 =& E(x,\pm ) - s T(x,\pm) + o(s),
\end{split}     
 \end{equation}
where 
\begin{equation}
T(x,\pm):= \int_0^\infty dt ~t \Phi_\pm(x,t).
\end{equation}
The quantity $T(x,\pm)$ the mean exit time of a particle starting at position $x$ in the internal velocity state $\pm 1$, if we set the exit time to zero for trajectories that do not exit the half-line.\\

Collecting the terms of order $0$ in $s$ in Eq.
  (\ref{PhiTildeEq})  yields a coupled system of evolution equations for the 
   exit probability in finite time, given the initial velocity state:
\begin{equation}\label{sysPhiABLap}
\begin{split}
 0 =& (F(x) + 1) \frac{\partial  E(x,+)}{\partial x}
   +   E(x,- ) -   E(x,+ ),\\
  0 =& (F(x) - 1) \frac{ \partial  E(x,- )}{\partial x}
   - E(x,- ) +   E(x,+ ).
\end{split}    
\end{equation}

Collecting the terms of order $1$ in $s$
  yields a coupled system of evolution equations 
   for the mean exit times:
\begin{equation}\label{sysTAB}
\begin{split}
 -E(x,+) =& (F(x) + 1) \frac{\partial T(x,+)}{\partial x}
   +  T(x,-) -  T(x,+),\\
 -E(x,-) =& (F(x) - 1) \frac{\partial 
T(x,-  )}{\partial x}
   -  T(x,-)  +  T(x,+).\\
\end{split}    
\end{equation}   
 These are the two systems of equations (Eqs (\ref{sysE},\ref{sysT})) announced in the introduction.\\

\section{Exit probability in a periodic force field with subcritical values}\label{ExitProbability}
\subsection{Integration of the evolution equations}
 
In our system of units, the evolution equation for the exit probability (Eq. (\ref{sysE})) becomes 
\begin{equation}\label{sysExit}
\begin{split}
 0 =& (F(x) + 1) \frac{\partial  E(x,+)}{\partial x}
   +  E(x,- ) -  E(x,+ ),\\
  0 =& (F(x) - 1) \frac{ \partial  E(x,- )}{\partial x}
   - E(x,- ) +  E(x,+ ).\\
\end{split}    
\end{equation}
Let us introduce the linear combinations, borrowing again the notations of \cite{gueneau2024relating}:
\begin{equation}\label{eEDef}
\begin{split}
E(x) := \frac{1}{2}\left(  E(x,+ ) + E(x,- )   \right),\\
e(x) := \frac{1}{2}\left(  E(x,+ ) - E(x,- )   \right).\\
\end{split}
\end{equation}
In terms of these unknowns, Eq. (\ref{sysExit}) becomes
\begin{equation}\label{sysHalfLine}
\begin{split}
 0 =& F(x) E'(x)
   +  e'(x),\\
  0 =& F(x) e'(x) + E'(x) - 2e(x).\\
\end{split}    
\end{equation}
Multiplying the second equation in the system by $F$ and substituting yields a differential equation of order one in $e$,
\begin{equation}\label{evole}
0 =  e'(x) + \frac{2F(x)}{ 1- F(x)^2 } e(x).
\end{equation}
 Integrating  yields
\begin{equation}\label{eExpr}
 e(x) = e(0)\exp(-J(x)),\quad (x\geq 0),
\end{equation}
with the notation
\begin{equation}\label{defJ}
\begin{split}
J(x) :=& \int_{0}^x \frac{ 2 F(u)  }{1 - F(u)^2}du.
\end{split}
\end{equation}
{\textcolor{black}{There is no zero in the denominator because we assume the values of the force field to be subcritical (see Eq. (\ref{eqn:<v}), which becomes $|F(x)|< 1$ in the present system of units where the parameter $v$ serves as the unit of velocity).}}\\

 Substituting into Eq. (\ref{sysHalfLine}) we obtain
\begin{equation}\label{diffEeSol}
\begin{split}
  E'(x) =& - \frac{2 e(0)}{F(x)^2 - 1}\exp\left( -J( x) \right),\\
\end{split}
\end{equation}
which upon integration yields
\begin{equation}\label{Eexpr}
\begin{split}
 E(x) =& E(0) + 2e(0) \Xi_-(x),\quad (x\geq 0),\\
\end{split}
\end{equation}
with the notation\footnote{The function  $\Xi_+$ will appear in Section \ref{MFPT}, in Eq. (\ref{solt}).}
\begin{equation}\label{defXi}
\begin{split}
 \Xi_\pm(x) :=& \int_0^x \frac{e^{\pm J(y)}}{1 - F(y)^2}dy.\\
\end{split}
\end{equation}

\subsection{Boundary conditions and return-time to the origin}
 We need two boundary conditions to fix the constants $E(0)$ and $e(0)$. A particle starting at the origin in a negative internal velocity state leaves the system immediately. Indeed $F(0)-1<0$ because the values of $F$ are subcritical. Hence 
 \begin{equation}
 E(0,-) = 1,
 \end{equation}
 which in terms of the integration constants $e(0)$ and $E(0)$ reads
  \begin{equation}\label{BC1}
  E(0) - e(0) = 1.
  \end{equation}

  As the force field $F$ is a periodic function of period $a$, the expressions of  $e(x)$ and $J(x)$ in Eqs (\ref{defJ},\ref{eExpr}) induce
\begin{equation}
 e( x + Na ) = e^{-NJ(a)} e(x),\;\;\;\;\;\;(x\geq 0, N\in\mathbb{N}).
\end{equation}
 The function $e$ is a linear combination of probabilities. It is therefore bounded. If $J(a)<0$, this implies that $e(x)$ is identically zero, and $e(0) = 0$. In this case, the expression of $E(x)$ in Eq. (\ref{Eexpr}) implies that the function $E$ is a constant. This constant is equal to $1$ due to the boundary condition of Eq. (\ref{BC1}). Hence, the particle exits the system almost surely if $J(a)<0$.\\ 

  If $J(a)=0$, then $J(ka)=0$ for any integer $k\geq 1$. Then, using the periodicity of $F$ in Eq. \eqref{defXi}, 
 \begin{equation}
 \Xi_{-}(ka) = \sum_{j=0}^{k-1} 
 \int_{ja}^{(j+1)a} \frac{e^{-J(y)}}{1 - F(y)^2}dy  =  k \int_0^a \frac{e^{-J(y)}}{1 - F(y)^2}dy\xrightarrow[]{k\to\infty}\infty.
 \end{equation}
 Using Eq. (\ref{Eexpr}), we obtain the sequence $E(ka) = E(0) + 2e(0) \Xi_-(ka)$ for any $k\geq 1$. As the function $E$ is bounded, this sequence is bounded and $e(0)$ must be equal to zero. The boundary condition of Eq. (\ref{BC1})  again implies that the RTP exits the positive half-line almost surely. Hence, {\textcolor{black}{the RTP exits almost surely if $J(a) \leq 0$, as announced in Eq. (\ref{condExit})}}.\\

 From now on let us assume that $J(a)> 0$.
 {\textcolor{black}{Consider a particle starting  at position $a$ in a positive internal velocity state}}  (its position at time $0$ is $a$, and its velocity is $F(a)+1>0$). To exit the positive half-line, the particle must come back to $a$. It does so with probability $E(0,+)$ because the periodic system on the half line $[a,\infty[$ is the same as on the half line $[0,\infty[$. When the particle comes back to $a$, it is in a negative internal velocity state. The exit probability from the system therefore satisfies 
 \begin{equation}\label{BCa}
    E( a,+) = E(0,+) E(a,-). 
 \end{equation}

In terms of the unknown functions $e$ and $E$, the above  condition reads
\begin{equation}\label{extro}
\begin{split}
 E(a) +e (a)  =& (E(0) + e(0) )(E(a) - e(a) ).\\
\end{split}
\end{equation}  
 Using the boundary condition of Eq. (\ref{BC1}) at the origin, and the expressions of $E(a)$  in terms of $E(0)$ and $e(0)$ in Eq. \eqref{Eexpr},
\begin{equation}
\begin{split}
E(0)=& 1 + e(0),\\
E( a ) =& E( 0 ) + 2 e(0) \Xi_-(a ),
\end{split}
\end{equation}
 the boundary condition of Eq. \eqref{extro} becomes  an equation in $e(0)$
 \begin{equation}
 \begin{split}
     1 + e(0)( 1 + 2\Xi_-(a) + e^{-J(a)}) =& [ 1 + 2e(0)][ 1 + e(0)( 1 + 2\Xi_-(a) - e^{-J(a)})],\\
     e(0)( -2 +2 e^{-J(a)}) =&  2e(0)^2 ( 1 - e^{-J(a)} + 2\Xi_-(a) ).
 \end{split}
 \end{equation}

There are two solutions to the above equation for $e(0)$. One is $0$ (if $e(0)=0$, we know that the particle exits the system in finite time almost surely). The other solution is negative. Let us show that under the assumption $J(a)>0$ the integration constant $e(0)$ is equal to the negative solution.\\

\begin{theorem}
If $J(a)>0$, 
\begin{equation}\label{e0Sol}
e(0) = -\frac{1-e^{-J(a)}}{1 -e^{-J(a)} + 2\Xi_-(a)}.
\end{equation}
\end{theorem}
\begin{proof}
We just need to show that if $J(a)>0$ and the particle starts at the origin with a positive velocity,  then the particle stays forever in the positive half-line with a strictly positive probability.  The trajectory of the particle is continuous, but we can map the trajectory   of an RTP on a real line to a Markov chain (as if the particle the particle lived on the points $\{ka\}_{k\in \mathbb Z}$).  If a particle is at position $ka$ in a positive velocity state, then it will either return to $ka$ with negative velocity or move to $(k+1)a$ with positive velocity, before getting out of the interval $[ka, (k+1)a]$; if the particle is at $ka$ in a negative velocity, then it either returns to $ka$ with the positive velocity or move to $(k-1)a$ with negative velocity, before getting out of the interval  $[(k-1)a, ka]$. These alternatives allow  us to introduce a Markov chain $(X_n)_{n\geq1}$ such that $X_n$ takes value either $1$ or $-1$.\\ 

\textcolor{black}{ The construction is as follows. \begin{itemize}
 \item Step 1: Assume that the particle is at position $0$ at time zero (it can be either in a positive or negative velocity state). We let $X_1=1$ (resp.\ $X_1=-1$) if the velocity is positive (resp.\ negative). 
 \item Step 2: The particle may move to $a$ (resp.\ $-a$) or $0$ if $X_1=1$ (resp.\ $X_1=-1$). 
 Then we let $X_2=1$ (resp.\ $X_2=-1$) if the velocity is positive (resp.\ negative) when the particle moves to $a, -a$ or $0.$
\item Step $n$: At the $n$-th step, the particle is at a location which is an integer multiple of $a$. We let $X_n=1$ (resp.\ $X_n=-1$) if the velocity is positive (resp.\ negative)  \end{itemize}
Therefore, we obtain a discrete-time process $(X_n)_{n\geq 1}$.
Since $F$ is periodic, $(X_n)_{n\geq 1}$ is a Markov chain. Note that although $X_n$ takes values $1$ or $-1$, $(X_n)_{n\geq 1}$ completely describes the trajectory of the particle on the points $\{ka\}_{k\in\mathbb Z}$. The distribution of the Markov chain $(X_n)_{n\geq 1}$ is completely characterised by its transition matrix which consists of $p_{i,j}$ for all $i,j=1,-1$, where
\[p_{1,1}:=\mathbb P(X_{n+1}=1\,|\, X_n=1), \quad p_{1,-1}:=1-p_{1,1}=\mathbb P(X_{n+1}=-1\,|\, X_n=1), \]
and 
\[p_{-1,1}:=\mathbb P(X_{n+1}=1\,|\, X_n=-1), \quad p_{-1,-1}:=1-p_{-1,1}=\mathbb P(X_{n+1}=-1\,|\, X_n=-1). \]
  Note that $p_{i,j}>0$ for any $i,j=1,-1$, thanks to Eqs  \eqref{eqn:<v}, \eqref{eqn:t+} and \eqref{eqn:t-}. Therefore, $(X_n)_{n\geq }$ is a finite-state irreducible Markov chain with a unique invariant measure, which we denote by $\pi=(\pi_1,\pi_2).$}\\

By this correspondence between the trajectory (living on the real line) and the Markov chain (living on $\{ka\}_{k\in\mathbb Z}$), we can express the following event in terms of the Markov chain:
\begin{equation}\label{eqn:equivalent}\begin{split}&\{\text{The particle stays forever in the positive half line starting from zero with positive velocity}\}\\
=&\left\{\forall k\geq 1 , \sum_{i=1}^kX_i>0\,\Big|\, X_1=1\right\}.\end{split}\end{equation}
We just need to show that 
\begin{equation}\label{eqn:forallk}
P\left(\forall k\geq 1 , \sum_{i=1}^kX_i>0\,\Big|\, X_1=1\right)>0.\end{equation}

Let us calculate the values $p_{i,j}$ and $\pi_1,\pi_2$ for the Markov chain $(X_n)_{n\geq 1}$. Let \[Z:=A+B+1,\]
with 
\begin{equation}\label{eqn:ab}
A=2 \int_0^a\Xi_-(x)dx,\quad B=e^{-J(a)}.
\end{equation}
 In \cite{gueneau2024relating}, the exit probability of an RTP from a segment has been worked out (in units of space and time where the velocity of the particle is $v_0$ and the tumbling rate is $\gamma$).
  The results of \cite[(16)]{gueneau2024relating} 
   yield the entries of the transition matrix of the Markov chain in our notations as 
\begin{equation}
p_{1,1}=\frac{2}{Z},\quad p_{1,-1}=\frac{Z-2}{Z},\quad p_{-1,1}=\frac{A-B+1}{Z},\quad p_{-1,-1}=\frac{2B}{Z}.
\end{equation}
The invariant measure $\pi$ satisfies 
\[(\pi_1,\pi_2)
\begin{pmatrix}
p_{1,1} & p_{1,-1}\\
p_{-1,1} & p_{-1,-1}
\end{pmatrix}=(\pi_1,\pi_2).\]
The above equation yields 
\[\pi_1=\frac{A-B+1}{2A},\quad \pi_2=\frac{A+B-1}{2A}.\]
Since $J(a)>0$, we have $B=e^{-J(a)}<1$. Then $\pi_1>\pi_2$. By the Ergodic Theorem for finite-state irreducible Markov chains \cite[Theorem 1.10.2]{norris1998markov}, 
\[\frac{\sum_{i=1}^nX_i}{n}\xrightarrow[]{n\to\infty}\pi_1-\pi_2>0, \quad \text{almost surely.}\]
for any initial distribution of $X_1$. The above display implies that $\sum_{i=1}^kX_i=0$ can occur only for a (random) finite number of values of the integer $k$. Therefore, \eqref{eqn:forallk} must hold,  which concludes the proof. 
\end{proof}

{\bf{Remark.}} The results of \cite{gueneau2024relating} for the exit probability of an RTP through any end of a segment $[a,b]$ have been used to express the transition matrix. 
 We can also consider the limit where $b$ goes to infinity (in the special case $a=0$), to obtain the above result. Indeed, if $E_b(x,\pm)$ denotes the exit probability from the segment $[0,b]$ through $b$, it has been shown in \cite{gueneau2024relating} that (in the present units and notations)
\begin{equation}
E_b(x, \pm) = \frac{\Xi_-(x)  \pm e^{-J(x)} + 1}{2\Xi_-(b) + e^{-J(b)} + 1}, \quad (x\in]a,b[).    
\end{equation}
 It is easy to use periodicity (see Eqs (\ref{JPer},\ref{XiPer}) in Appendix \ref{appInt}) to show that $J(a)>0$ implies 
\begin{equation}
 \underset{b\to \infty}{\lim}\Xi_-(b) = \frac{\Xi_-(a)}{1-e^{-J(a)}},\;\;\;\;\;\underset{b\to \infty}{\lim}J(b) = +\infty.
\end{equation} 
 We  obtain $e(0)$ by taking the following  large-$b$ limit:
\begin{equation}
\begin{split}
 e(0) =& \frac{1}{2}\left[\left(1- \underset{b\to \infty}{\lim} E_b(0,+) \right) - \left(1- \underset{b\to \infty}{\lim} E_b(0,-) \right)   \right]\\
=& -\frac{1}{ \frac{\Xi_-(a)}{1-e^{-J(a)}} + 1},
\end{split}
\end{equation}
 which gives back Eq. (\ref{e0Sol}).\\

 In the case $J(a)>0$, the coefficient $e(0)$ is therefore negative. Moreover, Eq.\ (\eqref{Eexpr})  implies that $E$ is a strictly decreasing function of the starting position $x$. In particular, the exit probability is not identically equal to $1$.\\

 To sum up, taking into account Eqs \eqref{eExpr}, \eqref{Eexpr} and \eqref{e0Sol}, the functions $e$ and $E$  are expressed as follows: 
\begin{equation}\label{sumUpE}
\begin{split}
 E(x ) =& \mathbbm{1}( J(a) \leq 0 ) + \left[ \frac{2\Xi_-(a) - 2(1-e^{-J(a)})\Xi_-(x)}{1-e^{-J(a)} + 2\Xi_-(a)} \right] \mathbbm{1}( J(a) > 0 ),\\
e(x) =&  -  \frac{1-e^{-J(a)}}{1-e^{-J(a)} + 2\Xi_-(a)}  e^{-J(x)}\mathbbm{1}( J(a) > 0 ).
\end{split}
\end{equation}
 From the definition of the unknowns in Eq. \eqref{eEDef}, the exit probability given the initial internal velocity state is expressed as follows: 
\begin{equation}\label{EExpr}
E( x,\pm) = \mathbbm{1}( J(a) \leq 0 ) + \left[ 1 - \frac{(1-e^{-J(a)})[1+ 2\Xi_-(x) \pm e^{-J(x)}]}{1-e^{-J(a)} +2\Xi_-(a)}\right] \mathbbm{1}( J(a) > 0 ).
\end{equation}

 The functions $J$ and $\Xi_-$ are defined (in Eqs \ref{defJ},\ref{defXi}) by integrals that need only be evaluated in the interval $[0,a[$. Indeed periodicity can then be used to work out the value of the exit probability at any point on the positive half-line. Substituting the expressions derived in the Appendix (Eqs (\ref{JPer},\ref{XiPer})) yields
\begin{equation}\label{EqNapm}
\begin{split}
E&( x+Na,\pm) = \mathbbm{1}( J(a) \leq 0 ) \\
&+ \left[ 1 - \frac{ 2 \Xi_-(a) + 1- e^{-J(a)}
 + e^{-NJ(a)}[(1-e^{-J(a)})(2\Xi_-(x)\pm e^{-J(x)}) -2\Xi_-(a)] }{1-e^{-J(a)} + 2\Xi_-(a)}\right] \mathbbm{1}( J(a) > 0 )\\
&= \mathbbm{1}( J(a) \leq 0 ) + \left[  \frac{
  e^{-NJ(a)}[2\Xi_-(a) - (1-e^{-J(a)})(2\Xi_-(x)\pm e^{-J(x)}) ] }{1-e^{-J(a)} + 2\Xi_-(a)}\right] \mathbbm{1}( J(a) > 0 ),\\
 & \shoveright{\quad(x\in [0,a[,\;\;\;\; N \in\mathbb{N}).}
\end{split}
\end{equation}
 It is easy to check that $E(x,\pm)$ goes to $E(a,\pm)$ when $x$ goes to $a$. 
 The second term corresponds to the case of non-zero survival probability. In such a case, the exit probability $E(Na,\pm)$ is an exponentially decreasing function of $N$. The exponential dependence of the exit probability on the distance to the origin {\textcolor{black}{in the case of a positive constant drift is recovered if $F$ is constant and positive  (this dependence is predicted by Eq. (\ref{EConst}) and checked as a special case  in Appendix \ref{checkConst}, see Eq. (\ref{repro}))}}.\\
 

\begin{figure}
\centering
    \includegraphics[width=0.9\linewidth]{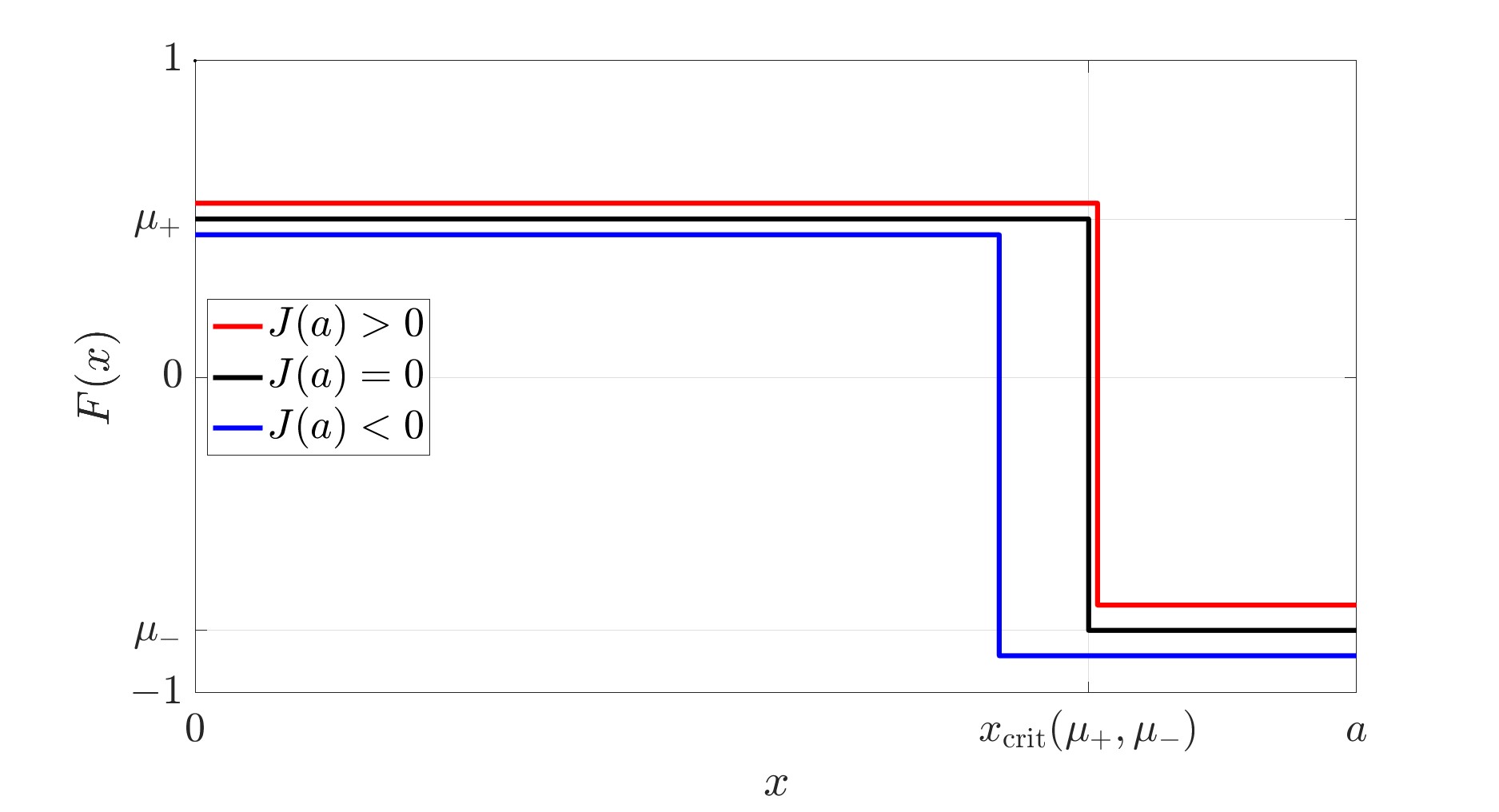}
    \caption{Piecewise-constant force fields alternating between two constant values of opposite signs. The black curve corresponds to the force field described in Eq. (\ref{exJ}). The value $x_{\mathrm{crit}}(\mu_-,\mu_+)$ is defined in terms of $\mu_-$ and $\mu_+$ in Eq. (\ref{xCritDef}). If the graph of such a piecewise-constant function does not intersect the black graph, the of $J(a)$ is positive (resp. negative) if the graph is above (resp. below) the black curve.}
    \label{figJ}
\end{figure}

{\textcolor{black}{
{\bf{Examples.}} Let us work out a few examples. \\ 
$(i)$ {\bf{Alternating drifts.}} Consider a force field taking the values $\mu_+$ and $\mu_-$, satisfying the conditions $-1<\mu_-<0<\mu_+ < 1$:
\begin{equation}\label{exJ}  
\begin{split}
 F( x + Na) =& \mu_+ \mathbbm{1}( 0\leq x < b) -\mu \mathbbm{1}( b\leq x < a),\;\;\;\;\;\;\;(x\in [0,a[,\;\;\;\;\;N\in\mathbb{N}),
\end{split}    
\end{equation}
 where $b$ is a value in $]0,a[$.  For definiteness we chose $F(0) =\mu_+$ and $F(b) = \mu_-$, but the values of $F$ on the discrete set $\{0,b\}$ do not change the value of $J(a)$, which is calculated as
\begin{equation}
 J(a) = 2\left[ \left(  \frac{\mu_+}{1-\mu_+^2} - \frac{\mu_-}{1-\mu_-^2}\right)b + \frac{\mu_-}{1-\mu_-^2}a  \right].
\end{equation} 
 For fixed $\mu_-$ and $\mu_+$, the quantity $J(a)$ has the same sign as $b-x_{\mathrm{crit}}(\mu_-,\mu_+)$, where 
\begin{equation}\label{xCritDef}
 x_{\mathrm{crit}}(\mu_-,\mu_+) = -\frac{\mu_-( 1 - \mu_+^2) a}{\mu_+( 1 - \mu_-^2) - \mu_-( 1-\mu_+)^2}.
\end{equation}
This position goes to $a$ (resp. $0$) if $\mu_-$ goes to $-1$ (resp. $\mu_+$ goes to $1$), which is intuitive and corresponds to cases where some of the values of the force field get close to critical. Examples are illustrated on Fig. \ref{figJ}, and the special case where $\mu_- = - \mu_+$ will be studied in detail in Section \ref{piecewiseConst} (see Eq. (\ref{alternatingDef})). In that case the critical value of the position is the midpoint of the interval $[0,a]$ as $x_{{\mathrm{crit}}}(-\mu_+,\mu_+)=1/2$.\\ 
}}

{\textcolor{black}{
$(ii)$ {\bf{Oscillations about a drift.}} Consider a force field oscillating about a subcritical value $\mu$, with an amplitude $\tau$:
\begin{equation}
F(x) = \mu + \tau \sin\left( \frac{2\pi x}{a}\right),
\end{equation}
with the assumptions $|\mu + \tau|< 1$ and $|\mu-\tau|<1$. 
 We restrict ourselves to a spatial frequency of $a^{-1}$. 
The sign of the integral $J(a)$ is unchanged if the frequency of the oscillation is multiplied by an integer.
 $J(a)$ is evaluated using the identity $\sin( y ) = 2\tan( y/2)( 1 + \tan( y/2)^2 )^{-1}$:
\begin{equation}
\begin{split}
 J( a ) =& 2\int_0^{a} \frac{ \mu + \tau \sin\left( \frac{2\pi x}{a}\right)}{1 - \left[ \mu + \tau \sin\left( \frac{2\pi x}{a}\right)  \right]^2}\\   
 =& \frac{a}{2\pi} \int_{-\pi}^\pi \frac{2( \mu + \tau \sin y )}{\left( 1 - \mu - \tau \sin y\right)\left( 1 +\mu + \tau \sin y\right)}dy\\
 =&  \frac{a}{2\pi}\left[\int_{-\pi}^\pi 
 \frac{ \mu + \tau \sin y }{1 - \mu - \tau \sin y}dy
  + \int_{-\pi}^\pi 
 \frac{ \mu + \tau \sin y }{1 + \mu + \tau \sin y}dy
 \right]\\
 =& \frac{a}{2\pi}\left[\int_{-\pi}^\pi 
 \frac{ dy }{1 - \mu - \tau \sin y}
  - \int_{-\pi}^\pi 
 \frac{ dy }{1 + \mu + \tau \sin y}
 \right].\\
\end{split}
\end{equation}
Using the integrals
\begin{equation}
\begin{split}
\int_{0}^X \frac{dz }{ 1 - a \sin z} =& \frac{2}{\sqrt{1-a^2}}\operatorname{arctan}\left( 
\frac{\tan\left(\frac{X}{2}  \right) - a}{\sqrt{1-a^2}}
\right),\\
\int_{-\pi}^{\pi} \frac{dz }{ 1 - a \sin z} =& \frac{2\pi}{\sqrt{1-a^2}},
\end{split}
\end{equation}
 we evaluate $J(a)$ as 
 \begin{equation}
J(a) = \left[ \frac{1}{\sqrt{(1-\mu)^2 - \tau^2}} -\frac{1}{\sqrt{(1+\mu)^2 - \tau^2}}   \right] a.
 \end{equation}
 The sign of $J(a)$ is therefore the sign of the average value $\mu$ of the oscillating force field. 
}}

\section{First-passage time at the origin}\label{MFPT}

Having worked out the expressions of $E(x,\pm)$, we can
 now solve the system of differential equations   satisfied by the mean first-passage times $T(x,\pm)$.
Changing unknowns to 
\begin{equation}\label{DefTt}
\begin{split}
 T(x) :=& \frac{1}{2}\left(  T(x,+) + T(x,-)  \right),\\
t(x) :=&  \frac{1}{2}\left(  T(x,+) - T(x,-)  \right),
\end{split}
\end{equation}
 Eq. (\ref{sysTAB}) yields a coupled system of differential equations:
\begin{equation}
\begin{split}
- E(x) =& F(x) T'(x) + t'(x),\\
- e(x) =& F(x) t'(x) + T'(x) -2t(x).\\
\end{split}    
\end{equation}
Combining and substituting in the same way as in Eq. (\ref{sysHalfLine}) of Section \ref{ExitProbability} yields an ordinary differential equation in $t$,
\begin{equation}\label{tDifGen}
t'(x) + \frac{2F(x)}{1-F(x)^2} t(x) = \frac{-E(x) + e(x) F(x)}{1- F(x)^2},
\end{equation} 
 and the expression of the derivative of $T$:
\begin{equation}\label{TPrime}
T'(x) = \frac{2t(x)}{1-F(x)^2} +\frac{-e(x) + E(x) F(x)}{1-F(x)^2}.
\end{equation}

\subsection{First-passage time at the origin with almost-sure exit ($J(a)\leq 0$)}

Assume $J(a) \leq 0$. In this case $E(x)$ is identically equal to 1 and $e(x)$ is identically equal to 0. Using the integrating factor $e^J$ we obtain (assuming $t(0)$ is finite),
\begin{equation}
 \frac{d}{dx}\left(   e^{J(x)} t(x)  \right) = -\frac{e^{J(x)}}{1-F(x)^2},
\end{equation}
\begin{equation}\label{solt}
\begin{split}
   t(x)  =& t(0)e^{-J(x)} - e^{-J(x)}\int_0^x \frac{e^{J(y)}}{1-F(y)^2}dy\\
   =& t(0)e^{-J(x)} - e^{-J(x)}\Xi_+(x).
\end{split}
\end{equation}


Integrating Eq. (\ref{TPrime}) therefore yields (assuming $T(0)$ is finite)
\begin{equation}\label{Tsolo}
\begin{split}
T(x) =& T( 0 ) + \int_0^x \frac{2t(v) + F(v)}{1-F(v)^2}dv\\
  =&  T(0)  + 2 t(0) \int_0^x \frac{e^{-J(v)}}{1-F(v)^2} dv - 
   2 \int_0^x \frac{e^{-J(v)}}{1-F(v)^2}\left[ \int_0^v \frac{e^{J(u) }}{1-F(u)^2}du \right] dv
 +\int_0^x \frac{F(v)}{1- F(v)^2} dv\\
 =& T(0)  + 2 t(0) \Xi_-(x) - 2 \int_0^x dv \Xi_-'(v) \Xi_+(v) dv + \frac{1}{2} J(x).
\end{split}
\end{equation}

 We have assumed that the two constants $T(0)$ and $t(0)$ are finite. To fix them, we need two boundary conditions. The particle leaves the system immediately if it starts at the origin with negative velocity:
\begin{equation}\label{BCOri}
 T(0,-) = 0.
\end{equation}
 A second boundary condition can be obtained from a renewal argument using the periodicity of the system. If the process starts at $a$ (at the end of a period interval of the field $F$), with a positive velocity, it enters the half-line $[a,\infty[$ and eventually returns to the origin (as $J(a)$ is negative, the particle reaches the origin almost surely). To reach the origin, it must first return to $a$ (which it does with a negative velocity). Moreover, the system in which the particle lives before its first return to $a$ is a half-line equipped with a periodic force field, identical to the entire system on the half-line $[0,\infty[$.\\

 Hence 
 the additional boundary condition
\begin{equation}\label{BCPer}
T(a,+) = T(0,+) + T(a,-).
\end{equation}
 In terms of the unknowns $T$ and $t$, 
 the boundary conditions of Eqs (\ref{BCOri},\ref{BCPer}) imply 
\begin{equation}\label{cocoinf}
T(0) = t(0) = t(a).
\end{equation}
 If we denote by $\tau_0$ (resp. $\tau_a$) the random time it takes for a particle starting at $0$ (resp. $a$) in a positive velocity state to return to $0$ (resp. $a$) {\textcolor{black}{in a negative velocity state.}} The first moment of $\tau_0$ is $t(0)$, and the first moment of $\tau_a$ is $t(a)$ (these moments may be infinite or finite). The two random variables $\tau_0$ and $\tau_a$ are identically distributed, hence they have the same first moment, which yields again $t(0)=t(a)$. Eq. (\ref{cocoinf}) therefore holds even if $t(0)$ and $t(a)$ are infinite. \\

 Let us distinguish two cases according to the value of $J(a)$.\\

$\bullet$ {\bf{Case 1: $J(a) = 0$.}} In this case, specializing Eq. (\ref{solt}) at $x=a$ yields
\begin{equation}
   t(a)  = t(0) - \int_0^a \frac{e^{J(y)}}{1-F(y)^2}dy.\\
\end{equation}
The second term on the r.h.s. is negative, hence the equation $t(a)=t(0)$ cannot be satisfied, which contradicts the assumption of a finite value of $t(0)$. Hence  
\begin{equation}\label{t0Inf}
    t(0)=\infty.
\end{equation} 
 This is consistent with the particular case of a free particle (where $F$ is identically zero, and $J(a)=0$). The mean first-passage time of a free RTP starting from the origin in a positive velocity state is known to be infinite \cite{de2021survival}.\\

$\bullet$ {\bf{Case 2: $J(a) <0$.}} Solving the equation $t(a) =t(0)$ using the expression of the function $t$ given in Eq. (\ref{solt}) yields  
\begin{equation}\label{t0Sure}
 t(0) = - \frac{e^{-J(a)}}{1- e^{-J(a)}}\Xi_+(a).
\end{equation}

 Substituting this integration constant into Eqs (\ref{solt},\ref{Tsolo}) yields the expressions 
\begin{equation}
\begin{split}
t(x) =& -e^{-J(x)}\left[ \frac{e^{-J(a)}}{1- e^{-J(a)}}\Xi_+(a) + \Xi_+(x)  \right],\\
T(x) = & -\frac{e^{-J(a)}}{1- e^{-J(a)}}\Xi_+(a)\left[ 1 + 2\Xi_-(x)\right] - 2 \int_0^x dv \Xi_-'(v) \Xi_+(v) dv + \frac{1}{2} J(x).
\end{split}
\end{equation}
The mean first-passage times given the initial velocity state are therefore obtained from Eq. (\ref{DefTt}) as
\begin{equation}\label{Txpm}
\begin{split}
T(x,\pm) =&  \frac{e^{-J(a)}}{1- e^{-J(a)}}\Xi_+(a) \left[ -1 - 2\Xi_-(x) \mp e^{-J(x)}      \right] \\
 &- 2 \int_0^x dv \Xi_-'(v) \Xi_+(v)  + \frac{1}{2} J(x) \mp e^{-J(x)} \Xi_+(x),\qquad(x\geq 0).\\
\end{split}    
\end{equation}

 Moreover, using the periodicity of $F$, it is enough to evaluate the integral expressions on the r.h.s. of the above equations for $x$ in the period interval $[0,a[$. The expressions of the various terms in the above equation at $x+Na$ for $x$ in $[0,a[$ an a positive integer $N$ are obtained in Appendix \ref{appInt} using periodicity\footnote{The calculations worked out in Appendix \ref{appInt} can be bypassed using a renewal argument, which yields Eq. (\ref{resT1}) and applies to the case $J(a)<0$ as well. The argument is presented in Section \ref{sectionCond}, Eq. (\ref{renewalPer}).} (Eqs (\ref{JPer},\ref{XiPer},\ref{intXiPrimeXiMP})). Substitution yields:
\begin{equation}\label{TExplAllx}
\begin{split}
T(x+Na,\pm) =&  \frac{e^{-J(a)}}{1- e^{-J(a)}}\Xi_+(a) \left[ -1 \mp e^{-NJ(a)}e^{-J(x)}
-2\frac{1- e^{-NJ(a)}}{1 - e^{-J(a)}}\Xi_-( a) -2 e^{-NJ(a)}\Xi_-(x)
\right] \\
 &- 2 \left[
\frac{ \Xi_+( a) \Xi_-(a)}{1 - e^{J(a)}}\left( -N + \frac{1 - e^{-NJ(a)}}{1- e^{-J(a)}} \right)
  + N \int_0^a\Xi'_-(v)\Xi_+(v) dv\right.\\
  & \left.- \frac{1- e^{-NJ(a)}}{1 - e^{J(a)}}\Xi_+( a)\Xi_-(x)  + 
   \int_0^x\Xi'_-(v)\Xi_+(v) dv
 \right]\\
 &+ \frac{1}{2}\left( NJ(a) + J(x) \right) \mp e^{-NJ(a) - J(x)} 
 \left[\frac{1- e^{NJ(a)}}{1 - e^{J(a)}}\Xi_+( a) + e^{NJ(a)}\Xi_+(x)\right],\\
 &\qquad\qquad(x\in [0,a[, N\in \mathbb{N}).
\end{split}    
\end{equation}
 Collecting the terms according to their dependence on $N$ yields only terms proportional to $N$ and $1$ (the terms proportional to $e^{-NJ(a)}$ vanish):
\begin{equation}\label{resT1}
\begin{split}
T(x+Na,\pm) =& \left(  \frac{J(a)}{2} +2\frac{ \Xi_+( a) \Xi_-(a)}{1 - e^{J(a)}} -2 \int_0^a dv\Xi'_-(v)\Xi_+(v) \right)N\\
&+ \frac{\Xi_+(a)}{1-e^{J(a)}}\left( 1 \pm e^{-J(x)} + 2\Xi_-(x)\right)
 +\frac{1}{2}J(x) \mp e^{-J(x)}\Xi_+(x) - 2 \int_0^x dv\Xi_-'(v)\Xi_+(v),\\
 &\qquad\qquad(x\in [0,a[, N\in \mathbb{N}).
\end{split}
\end{equation} 
 This is the result reported in Eq. (\ref{reportedTNeg}). The affine dependence of $T(x+Na,\pm)$ in $N$ generalizes the  behavior identified for a force field consisting of a constant negative drift (this behavior corresponds to $\mu<0$ in the statement $(\mathrm{II})_\mu$, Eq. (\ref{resIImu})). Consistency is checked in Appendix \ref{constantCheckNeg}.\\



\subsection{First-passage time at the origin with non-zero survival probability ($J(a)>0$)}
\subsubsection{Integration of the evolution equation}
 In this case we need to solve the more general  differential equations (Eqs (\ref{tDifGen},\ref{TPrime})), substituting the expressions of $e(x)$ and $E(x)$ obtained in Eqs (\ref{eExpr},\ref{Eexpr}):
\begin{equation}
\begin{split}
e(x) =& e( 0 ) e^{-J(x)},\\
E(x) =& 1+e(0) + 2 e(0) \Xi_-(x),\\
 e(0) =& -\frac{1-e^{-J(a)}}{1-e^{-J(a)}  + 2\Xi_-(a)}.
\end{split}
\end{equation}
Using the same integrating factor $e^{J}$ as in Eq. (\ref{solt}) and substituting yields
  \begin{equation}\label{txExpl}
\begin{split}
t(x) =& t(0)e^{-J(x)}  + \int_0^x \frac{-E(y) + e(y) F(y)}{1-F(y)^2} e^{J(y) - J(x)} dy,\\
\end{split}
\end{equation}

\begin{equation}\label{TxExpl}
\begin{split}
T(x) =& T(0) + 2t(0) \int_0^x \frac{e^{-J(y)}}{ 1 - F(y)^2} dy\\
&+ 2 \int_0^x dz \frac{e^{-J(z)}}{1-F(z)^2}\int_0^z \frac{e^{J(y)}[-E(y) + e(y) F(y)]}{1-F(y)^2} dy
+\int_0^x \frac{E(v)F(v) -e(v)}{1- F(v)^2} dv.\\
\end{split}
\end{equation}

The functions $t$ and $T$ are therefore expressed in terms of the initial values $t(0)$ and $T(0)$ as
\begin{equation}\label{tTExpr}
\begin{split}
  t( x )  =& t(0) e^{-J(x)} +\varphi(x),\\
  T(x )  =&  T(0) + 2 t(0) \Xi_-(x) +  \psi(x),\\
\end{split}
\end{equation}
 where the two functions $\varphi$ and $\psi$ do not depend on the integration constants $t(0)$ and $T(0)$. They are expressed in terms of the force field and of the functions $e$ and $E$ (expressed in Eqs (\ref{eExpr}, \ref{Eexpr})) as follows:
\begin{equation}\label{phiPsiDef}
\begin{split}
\varphi(x):=& e^{-J(x)} \int_0^x\frac{e(y)F(y) - E(y)}{1-F(y)^2}e^{+J(y)} dy,  \\
\psi(x) :=& 2\int_0^x \frac{\varphi(y)}{1-F(y)^2} dy + \int_0^x\frac{E(y) F(y) -e(y)}{1-F(y)^2}dy.\\
\end{split}
\end{equation}


\subsubsection{Boundary conditions and return-time to the origin}
Consider an initial state with positive velocity at position $a$, with a force field $F$ satisfying the condition $J(a)>0$. To reach the origin in finite time, it needs to come back to position $a$. It does so with probability $E(0,+)$, because the half-line $[a,\infty[$ is endowed with the same force field as $[0,\infty[$. If it comes back to $a$, it passes through $a$ with negative velocity and reaches the origin with probability $E(a,-)$.\\

 Consider the conditional average $\langle T(x) \rangle_{c,\epsilon}$ of the first-passage time at the origin, defined in Eq. (\ref{condDef}). 
 The above argument implies that the conditional average of the first-passage time of the particle starting at $a$ in the positive velocity state satisfies
\begin{equation}\label{beco}
\langle T(a) \rangle_{c,+} E( a,+) =  \langle T(0) \rangle_{c,+} E(0,+) \times E( a,-)   +  E( 0,+)
\langle T(a) \rangle_{c,-} E( a,-) .
\end{equation}
This condition can be re-expressed in terms of the solutions of the evolution equations Eq. (\ref{sysT}) as 
\begin{equation}\label{eqn:T=1+1}
  T(a,+) = T( 0,+) E( a,-) +  T(a,- ) E( 0,+),
\end{equation}
 On the other hand, a particle leaves the system immediately if it is at the origin at time $0$ in a negative velocity state, meaning $T(0,-) = 0$, hence 
\begin{equation}
    t(0) = T(0) = \frac{1}{2}T(0,+).
\end{equation}

 Substituting into Eq. (\ref{eqn:T=1+1}) yields
\begin{equation}
T( a ) + t( a ) = 2t(0) E( a,-) + (T(a) - t(a)) E(0,+).
 \end{equation}
 Substituting the expressions of the functions $T(a)$ and $t(a)$ found in Eq. (\ref{tTExpr}), and using the identities $E(0,+) = 2e(0) + 1$ and $E(a,-) = e^{-J(a)}$ yields 
 \begin{equation}
t( 0 ) \times (-2e(0))[ 1 + 2\Xi_-(a) - e^{-J(a)}] = 2e(0) \psi(a) -2[ e(0) + 1]\varphi(a),\\ 
\end{equation}
 Substituting the expression of $e(0)$ in Eq. (\ref{e0Sol}) yields the integration constant
\begin{equation}\label{t0Expr}
t(0)= \frac{  e(0) \psi(a) -[ e(0) + 1] \varphi(a)   }{ 1 - e^{-J(a)}}= \frac{-\psi(a) - \frac{2\Xi_-(a)}{1 - e^{-J(a)}}\varphi(a)}{1-e^{-J(a)} + 2\Xi_-(a)},
\end{equation}
 The solution of Eq. (\ref{sysT}) follows:
\begin{equation}\label{TClosedAllx}
\begin{split}
T(x,\pm)=& t(0)[ 1 \pm e^{-J(x)} + 2\Xi_-(x)  ]
+\varphi(x) \pm \psi(x), \qquad( x\geq 0 ).
\end{split}
\end{equation}
 
 Using periodicity as in the case of almost-sure exit, we can express the functions $t$ and $T$ at position $x+Na$, with $x$ in $[0,a[$, to evaluate\footnote{As in the case of almost-sure exit, the calculations worked out in the Appendix can be bypassed, as the conditional average of the first-passage time reported in Eq. (\ref{reported}) can be obtained  using the renewal argument presented in Section \ref{sectionCond}, Eq. (\ref{renewalPer}).}
 \begin{equation}\label{TxNaSurv}
T(x+Na, \pm) = t(0) + 2t(0) \Xi_-( x+Na ) + \psi( x+Na) 
\pm t(0) e^{-NJ(a)} e^{-J(x)} \pm \varphi(x+Na).
 \end{equation}
 The calculations are done in Appendix \ref{periodApp}, they lead to Eq. \eqref{appTpmp},
  which yields:

\begin{equation}\label{Tpmp}
\begin{split}
T(x+Na, \pm )
=  &\varphi(a) e^{J(a)}\left( 2\Xi_-(x)-2 \frac{\Xi_-(a)}{1-e^{-J(a)}} \pm e^{-J(x) }\right)Ne^{-NJ(a)}\\ 
&+ \left[ \frac{(-\psi(a) - 2\frac{\Xi_-(a)}{1 - e^{-J(a)}}\varphi(a))( 1 \pm e^{-J(x) }+ 2\Xi_-(x) )}{1-e^{-J(a)} + 2\Xi_-(a)} +\psi(x) \pm\varphi(x)\right] e^{-NJ(a)},\\
&\quad(N\in \mathbbm{N}, x \in [0,a[). 
\end{split}
\end{equation}

 All terms in the expression of $T(x\pm Na)$ are weighted by a factor of $e^{-N J(a)}$, but this factor is present in $E(x+Na,\pm)$ (see Eq. (\ref{EqNapm})). The quantities $T(x+Na,\pm)$ therefore go to zero in the large-$N$ limit, but this is an indication of the low exit probability.\\

\section{Conditional average of the exit time over trajectories that reach the origin}\label{sectionCond}

The values of the exit probabilities $E(x+Na,\pm)$ and 
 of the times $T(x+Na,\pm)$ obtained in Eqs (\ref{EqNapm},\ref{Tpmp}) yield the conditional averages $\langle T(x+Na) \rangle_{c,\pm}$ in all cases.\\

\noindent{$\bullet$} {\bf{Almost-sure exit ($J(a)\leq 0$).}}
There are two subcases, studied in the previous two sections.\\

$(i)$ $J(a)=0$.\\
 In this case, the exit probability is equal to $1$, but the times $T(x,\pm)$ are not finite. Hence 
 \begin{equation}
 \langle T(x) \rangle_{c,\pm}=\infty, \qquad (x\geq 0).
 \end{equation}

$(ii)$  $J(a)<0$.\\ 
In this case, the conditional average of the first-passage time (with given initial velocity state $\pm 1$) at the origin coincides with $T(x,\pm)$, as reported in Eq. (\ref{resT1}):\\

\begin{equation}\label{resC1}
\begin{split}
\langle T(x+Na) \rangle_{c,\pm} =& 
 T(x+Na,\pm)\\
=&\left(  \frac{J(a)}{2} +2\frac{ \Xi_+( a) \Xi_-(a)}{1 - e^{J(a)}} -2 \int_0^a dv\Xi'_-(v)\Xi_+(v) \right)N\\
&+ \frac{\Xi_+(a)}{1-e^{J(a)}}\left( 1 \pm e^{-J(x)} + 2\Xi_-(x)\right)
 +\frac{1}{2}J(x) \mp e^{-J(x)}\Xi_+(x) - 2 \int_0^x dv\Xi_-'(v)\Xi_+(v),\\
 &\qquad\qquad(x\in [0,a[, N\in \mathbb{N}).
\end{split}
\end{equation} 

\noindent{$\bullet$} {\bf{Non-zero survival probability ($J(a)>0$).}} 
  Dividing the time $T(x+Na,\pm)$ (expressed in Eq. (\ref{appTpmp})) by the exit probability $E(x+Na,\pm)$  (expressed in Eq. (\ref{EqNapm})),  we  obtain $\langle T(x+Na)\rangle_{c,\pm}$ as 

\begin{equation}\label{resC2}
 \begin{split}
\langle T(x+Na)\rangle_{c,\pm}
=&  \frac{1-e^{-J(a)} + 2\Xi_-(a)}{2\Xi_-(a) - (1-e^{-J(a)})(2\Xi_-(x) \pm e^{-J(x)})  }\varphi(a) e^{J(a)}\left( 2\Xi_-(x)-2 \frac{\Xi_-(a)}{1-e^{-J(a)}} \pm e^{-J(x) }\right)N\\ &+ 
\frac{(-\psi(a) - 2\frac{\Xi_-(a)}{1 - e^{-J(a)}}\varphi(a))(1 \pm e^{-J(x)} + 2\Xi_-(x)) + ( \psi(x) \pm \varphi(x) )(1- e^{-J(a)}+ 2\Xi_-(a) ) }{2\Xi_-(a) - (1-e^{-J(a)})(2\Xi_-(x) \pm e^{-J(x)})  }
\\
=& -\frac{(1- e^{-J(a)} + 2\Xi_-(a))\varphi(a) e^{J(a)}}{1-e^{-J(a)}}N\\
 &+ 
\frac{(-\psi(a) - 2\frac{\Xi_-(a)}{1 - e^{-J(a)}}\varphi(a))(1 \pm e^{-J(x)} + 2\Xi_-(x)) + ( \psi(x) \pm \varphi(x) )(1- e^{-J(a)}+ 2\Xi_-(a) ) }{2\Xi_-(a) - (1-e^{-J(a)})(2\Xi_-(x) \pm e^{-J(x)})  },\\
&\qquad\qquad(x\in [0,a[, N\in \mathbb{N}).
\end{split}
\end{equation} 
The above expressions are the ones reported in Eqs (\ref{reportedTNeg},\ref{reported}) in the introduction. Indeed, the coefficient of $N$ can be rewritten using the expression of $e(0)$ in Eq. (\ref{e0Sol}) as 
\begin{equation}
\begin{split}
- \frac{1-e^{-J(a)} + 2\Xi_-(a)}{ 1 - e^{-J(a)} }\varphi(a) e^{J(a)} =& \frac{1}{e(0)}\int_0^a dy\left[ \frac{e(0)F(y)}{1- F(y)^2} - \frac{2\Xi_-(a) - 2 (1-e^{-J(a)})\Xi_-(y) }{1 - e^{-J(a)} + 2\Xi_-(a)}\times\frac{e^{J(y)}}{1- F(y)^2}\right]\\
=& \frac{J(a)}{2} + 2\frac{\Xi_-(a)\Xi_+(a)}{1-e^{-J(a)}} - 2 \int_0^a dv \Xi_-(v)\Xi_+'(v) dv.
\end{split}    
\end{equation}

 For any non-zero value of $J(a)$, the quantities $\langle T(x+Na )\rangle_{c,\pm}$ 
 are affine functions of $N$, and the coefficient of $N$ does not depend on $x$ or on the initial velocity state. This is observed directly after an explicit calculation using periodicity, completed separately in the cases $J(a)<0$ and $J(a)>0$. However, we can obtain $\langle T(x+Na)\rangle_{c,\pm}$    directly from a renewal argument similar to the one we used to work out the boundary conditions in Eqs (\ref{BCPer},\ref{eqn:T=1+1}).
Consider a nonnegative integer $N$ and $x$ in $[0,a[$. A particle starting  at coordinate $x+Na$ (in the internal velocity state $\pm$) must go through $Na$ with negative velocity to eventually reach the origin. It does go through $Na$ with probability $E(x,\pm)$ by periodicity of the system. After going through $Na$, it reaches the origin with probability $E(Na,-)$, and 
\begin{equation}\label{renewalPer}
\begin{split}
T( x + Na, + ) =&  T(x,+)E( Na, - ) 
 + T( Na, -) E( x,+),\\
T( x + Na, - ) =&  T(x,-)E( Na, - ) + T( Na, -) E( x,-).\\    
\end{split}
\end{equation}
In terms of the conditional averages, the above  conditions become 
\begin{equation}\label{renewalPerC}
\begin{split}
\langle T  (x + Na) \rangle_{c, +} E( x + Na, + ) =&  \langle T  (x )\rangle_{c,+}E(x,+)E( Na, - )\\ 
& +  \langle T  (Na)\rangle_{c, -} E(Na, - )E( x,+),\\
\langle T  ( x + Na)\rangle_{c,-} E( x + Na, - ) =&  \langle T  ( x)\rangle_{c,-} E(x,-) E( Na, - ) \\
&+ \langle T  ( Na)\rangle_{c,-} E( Na, - )E( x,-).\\   
\end{split}
\end{equation}
On the other hand, we can check from Eq. (\ref{EqNapm}) that
\begin{equation}
    E( x + Na, \pm ) = E(x,\pm) E( Na, - ).
\end{equation}
 Substituting into Eq. (\ref{renewalPerC}) yields
\begin{equation}
\begin{split}
\langle T  ( x + Na)\rangle_{c, + } =&  \langle T  (x)\rangle_{c,+} + \langle T (Na) \rangle_{c,-},\\
\langle T  ( x + Na)\rangle_{c, -} =&  \langle T  (x)\rangle_{c,-} + \langle T (Na) \rangle_{c,-}.\\   
\end{split}
\end{equation}
An immediate induction yields
\begin{equation}
\begin{split}
\langle T  ( x + Na)\rangle_{c, + }=&  \langle T  ( x )\rangle_{c, + } + N \langle T ( a)\rangle_{c, - } = \frac{T(x,+)}{E(x,+)} + N \frac{T(a,-)}{E(a,-)} ,\\
\langle T  ( x + Na)\rangle_{c, - }=&  \langle T  ( x )\rangle_{c, - } + N \langle T ( a)\rangle_{c, - } = \frac{T(x,-)}{E(x,-)} + N \frac{T(a,-)}{E(a,-)},\\
\qquad(x\in[0,a[, N \in \mathbbm{N} ).\\
\end{split}
\end{equation}
 The conditional average is therefore obtained on the entire positive half-line from the solution of Eqs (\ref{sysE},\ref{sysT}) on the period interval $[0,a]$. In the case of almost sure exit, the conditional averages $\langle T(x)\rangle_{c,\pm}$ reduce to $T(x,\pm)$, and Eq. (\ref{reportedTNeg}) follows from evaluating the expressions of $T(a,-)$ and $T(x,\pm)$ given in Eq. (\ref{TExplAllx}), for $x$ in the interval $[0,a[$. In the case of non-zero survival probability, evaluating $E(x,\pm)$ and $E(a,-)$ from Eq. (\ref{EqNapm}) and the quantities $T(x,\pm)$ and $T(a,-)$ from Eq. (\ref{TClosedAllx}) gives back Eq. (\ref{reported}).\\

\section{Example: alternating drifts}\label{piecewiseConst}

A simple periodic modification of the model with constant drift studied  in \cite{de2021survival} is the one-parameter family of periodically-alternating positive and negative drifts with the same amplitude. 
 Let us define the periodic field $F$ on the positive half-line by 
\begin{equation}\label{alternatingDef}
\begin{split}
 F( x + Na) =& \mu \mathbbm{1}( 0\leq x < a(1-\epsilon)) -\mu \mathbbm{1}( a(1-\epsilon)\leq x < a),\;\;\;\;\;\;\;(x\in [0,a[,\;\;\;\;\;N\in\mathbb{N}),
\end{split}    
\end{equation}
 {\textcolor{black}{where $\epsilon$ is in $[0,1]$, and $\mu$ is in $[0,1[$ (which ensures that the force field takes only subcritical values).}} With this definition, the parameter $J(a)$ is positive if $\epsilon$ is in $[0,\frac{1}{2}[$, negative if $\epsilon$ is in $]\frac{1}{2},1[$.
 An example is plotted in Fig. \ref{alternatingFig}. The limits $\epsilon\to 1$ and $\epsilon\to 0$ correspond respectively to the constant drifts $-\mu$ and $\mu$, which are worked out in Appendix \ref{checkConst}.\\

\begin{figure}
\centering
    \includegraphics[width=0.5\linewidth]{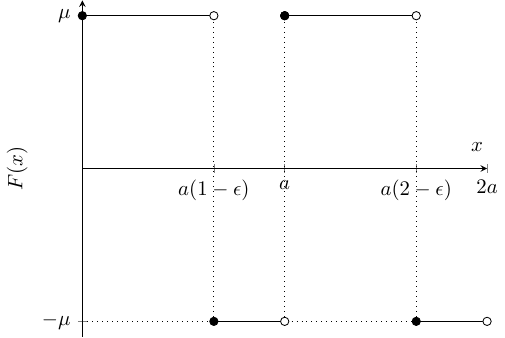}
    \caption{The periodic field $F$ alternating between $\mu$ and $-\mu$ defined in Eq. (\ref{alternatingDef}), with positive $\mu$. The quantity $J(a)$ has the same sign as $\epsilon - \frac{1}{2}$. On this graph, $\epsilon<\frac{1}{2}$ (which corresponds to a non-zero survival probability).}
    \label{alternatingFig}
\end{figure}

\subsection{Exit probability}

The function $J$ defined in Eq. (\ref{defJ}) is evaluated on $[0,a[$ as
\begin{equation}\label{JPiece}
\begin{split}
J(x) =& \alpha x \mathbbm{1}( x\in[0, a( 1-\epsilon)[) +
 \alpha( -x + 2a( 1 - \epsilon) ) \mathbbm{1}( x\in [ a( 1-\epsilon), a]),\;\;\;\;\;\;\quad(x\in [0,a]),\\ 
{\mathrm{with}}&\;\;\;\;\;\alpha:= \frac{2\mu}{1-\mu^2}.
\end{split} 
\end{equation}
In particular, 
\begin{equation}\label{JaVal}
 J( a ) =  \frac{2\mu}{1-\mu^2}a( 1 - 2\epsilon).
\end{equation}
The case of a constant drift corresponds to $\epsilon = 0$, and as expected $J(a)$ is of the sign of $\mu$ if $\epsilon<1/2$.\\

Evaluating the integral in the definition of $\Xi_\pm$ in Eq. (\ref{defXi}) yields
\begin{equation}\label{XiPiece}
\begin{split}
 \Xi_n(x) =& \frac{1}{2n\mu}\left( e^{n\alpha x} - 1\right) \mathbbm{1}( x\in [0, a( 1-\epsilon)[ )\\
&+\frac{1}{2n\mu}\left( 2 e^{n\alpha a(1-\epsilon)} - 1 - e^{2n\alpha a(1-\epsilon)}e^{-n\alpha x}   \right) 
 \mathbbm{1}( x\in[ a( 1-\epsilon), a]),\;\;\;\;\;\;(0\leq x \leq a, \quad n\in \{-,+\}).
 \end{split}
\end{equation}
In particular,
\begin{equation}\label{Xina}
\begin{split}
 \Xi_n(a) =& \frac{1}{2n\mu}\left[  e^{n\alpha a}( 2e^{-n\alpha a \epsilon} -e^{-2n\alpha a \epsilon}) - 1\right].
 \end{split}
\end{equation}

Substituting into Eq. (\ref{EqNapm}), the exit probability for a particle starting at the origin in a positive internal velocity state is obtained as a function of the parameter $\epsilon$:
\begin{equation}\label{exprEAlternating}
\begin{split}
E( 0, +) =&    \left[ \frac{1-\frac{1-e^{-J(a)}}{2\Xi_-(a)}}{1+\frac{1-e^{-J(a)}}{2\Xi_-(a)}}   \right] \mathbbm{1}\left( \epsilon < \frac{1}{2} \right) + \mathbbm{1}\left( \epsilon \geq \frac{1}{2} \right) \\
=&  \frac{1 - \mu\left( \frac{1-e^{-\alpha a(1-2\epsilon)}}{1 - e^{-\alpha a}( 2 e^{\alpha a\epsilon} - e^{2 \alpha a \epsilon}) }
    \right) }{ 1 +  \mu\left(   \frac{1-e^{-\alpha a(1-2\epsilon)}}{1-e^{-\alpha a}( 2 e^{\alpha a\epsilon} - e^{2 \alpha a \epsilon})}   \right)       }\mathbbm{1}\left( \epsilon < \frac{1}{2} \right) + \mathbbm{1}\left( \epsilon \geq \frac{1}{2} \right) .
 \end{split}
\end{equation}
This exit probability $E(0,+)$ is plotted as a function of $\epsilon$ in Fig. \ref{figureE0AlternatingDebug}, together with the results of numerical simulations. One notices that for a fixed value of $\mu$, the limit of small $a$  yields a function of $\epsilon$ only:
\begin{equation}\label{effectiveDrift}
 E(0,+)\underset{a\ll 1}{\sim} \frac{1-\mu( 1 - 2\epsilon)}{1+\mu( 1 + 2\epsilon)} \mathbbm{1}\left( \epsilon<\frac{1}{2}\right) + \mathbbm{1}\left( \epsilon \geq \frac{1}{2}\right).
\end{equation}
 Comparing with Eq. (\ref{EConst}), we recognize the value of the exit probability $E(0,+)$ in an effective constant drift $\mu( 1-2\epsilon)$. Indeed, this effective is of the sign of $1-2\epsilon$ as the parameter $\mu$ is positive. On the other hand, in the limit of large $a$, the exit probability goes to the value of $E(0,+)$ obtained in the case of a constant positive subcritical drift equal to $\mu$:
 \begin{equation}
     E(0,+)\underset{a\gg 1}{\sim} \frac{1-\mu}{1+\mu} \mathbbm{1}\left( \epsilon<\frac{1}{2}\right) + \mathbbm{1}\left( \epsilon \geq \frac{1}{2}\right).
 \end{equation}
 This is consistent with intuition: in the limit of large $a$, the particle is unlikely to explore a region of space where the drift is different from $\mu$.\\

\begin{figure}
\includegraphics[width=1.1\textwidth,keepaspectratio]{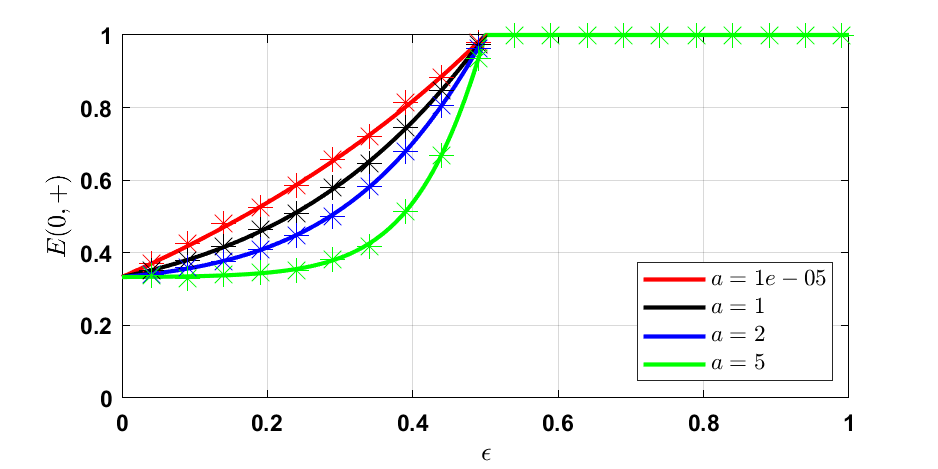}
\caption{The exit probability $E(0,+)$ for a particle starting at the origin in an internal positive velocity state, for a drift alternating between $\mu= 1/2$ on $[0,a(1-\epsilon)[$ and $-\mu$ on $[a(1-\epsilon), a[$, as a function of $\epsilon$ (given by Eq. (\ref{exprEAlternating})). For any value of the period $a$, the limit of $E(0,+)$ when $\epsilon$ goes to zero is $(1-\mu)/(1+\mu) = 1/3$. The star symbols are the results of numerical simulations of trajectories of particles starting from the origin with positive velocity. The simulation ended at time $100 t(0)$, where $t(0)$ is obtained in Eq. (\ref{t0Sure}) for $\epsilon>\frac{1}{2}$ and in Eq. (\ref{t0Expr}) for $\epsilon<\frac{1}{2}$. The simulated value of $E(0,+)$ is the share 
 of particles that return to the origin before the end of the simulation.}
\label{figureE0AlternatingDebug}
\end{figure}

 Once the exit probability given an initial position at the origin has been evaluated, the spatial dependence of the exit probability follows from Eq. (\ref{EqNapm}). 
The exit probability $E(y,\pm)$ is shown in Fig. \ref{figureExAlternating} as a function of $y$ for a period $a$ equal to the mean free path of the particle, for $\epsilon = 0.25$ (which corresponds to a non-zero survival probability).
\begin{figure}
\centering
\includegraphics[width=1.1\textwidth,keepaspectratio]{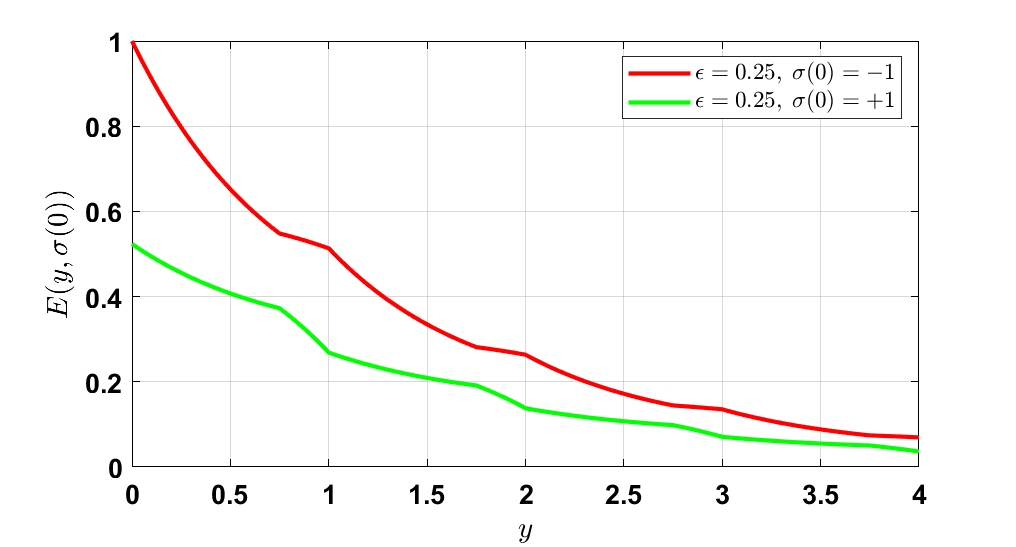}
    \caption{The exit probability as a function of the starting point from the origin. In this plot, $\mu = 1/2$, and $a = 1$.}
    \label{figureExAlternating}
\end{figure}

\subsection{First-passage time at the origin}

\subsubsection{Almost-sure exit with finite mean first-passage at the origin ($J(a)<0$)}

The particle exits the system almost surely if $\epsilon \geq \frac{1}{2}$. If $\epsilon = \frac{1}{2}$, the mean first-return time to the origin $t(0)$ is infinite.\\

Let us assume $\epsilon > \frac{1}{2}$, so that the parameter $J(a)$ is negative and $t(0)$ is given by Eq. (\ref{t0Sure}). The value of $\Xi_+(a)$ is obtained from Eq. (\ref{Xina}):
\begin{equation}\label{T0PlusSure}
\begin{split}
 T(0,+) = 2t( 0 ) =& - \frac{2}{e^{J(a)} - 1}\Xi_+(a)= \frac{2e^{\alpha a( 1 - \epsilon)} - e^{\alpha a( 1 - 2\epsilon)} - 1 }{\mu(1- e^{\alpha a ( 1 - 2\epsilon)})}.
\end{split}
\end{equation}

Once the conditional average of the mean first-return time to the origin has been evaluated, the spatial dependence of the mean first-passage time follows from the general expression obtained in Eq. (\ref{resC2}). The conditional averages $\langle T(y) \rangle_{c,\pm}$ are shown as a function of $y$ in Fig. \ref{figureCxAlternating}, where the plots with $\epsilon = 0.75$ correspond to the case of almost-sure exit.\\

\begin{figure}
    \centering    \includegraphics[width=1.1\textwidth,keepaspectratio]{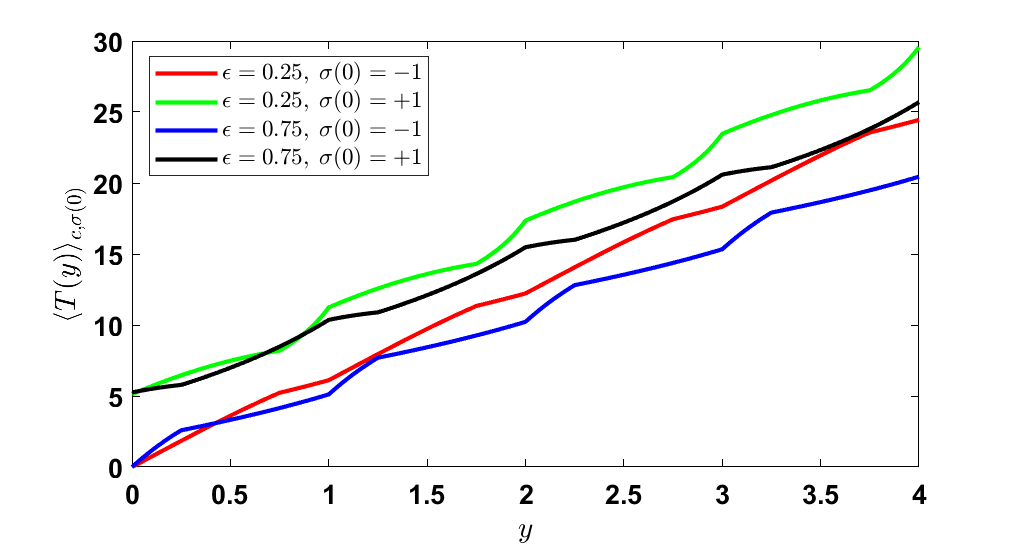}
    \caption{The conditional average of the first-passage time at the origin as a function of the starting point from the origin. In this plot, $\mu = 1/2$, and $a = 1$.}
    \label{figureCxAlternating}
\end{figure}


Let us work out the following limits of the mean first-return time to the origin.\\

$\bullet$ {\bf{Limit of constant negative drift ($\epsilon\to 1$).}} 
 Let us write  $\epsilon = 1 - h$, with small and positive $h$. When $h$ goes to zero, the field $F$ goes to the constant $-\mu$.  From Eq. (\ref{T0PlusSure}),
 the first-return time $T(0,+)$ goes to a finite limit when $h$ goes to zero:
\begin{equation}\label{T0PlusAs}
T(0,+) = \frac{2 e^{-\alpha a h} - e^{-\alpha a( 1 + 2h)} - 1}{ \mu(1 - e^{-\alpha a (1 + 2h)} )} = \frac{2 - e^{\alpha a} - 1}{\mu( 1 - e^{-\alpha a})}+o(1) = \frac{1}{\mu} + o(1).
\end{equation}
 This limit is the mean first-return time to the origin of an RTP with constant negative drift $-\mu$, as reported in Eq. (\ref{resIImu}).\\

$\bullet$ {\bf{Limit of small, negative $J(a)$ ($\epsilon\to \frac{1}{2}^+$).}} Let us write
\begin{equation}
\epsilon:= \frac{1}{2} + h,\;\;\;\;\;\; h>0.
\end{equation}
Substitution in Eq. (\ref{T0PlusSure}) and Taylor expansion yield
{\textcolor{black}{
\begin{equation}
T(0,+) = \frac{2 e^{\alpha a\left(\frac{1}{2} - h  \right)} - e^{-2\alpha a h} - 1}{\mu( 1- e^{-2\alpha a h})}
\underset{h\to 0^+}{\sim} \frac{2(e^{\frac{\alpha a}{2}} - 1)}{2\mu\alpha a h},\\
\end{equation}
hence
\begin{equation}
T(0,+)  \underset{\epsilon\to \frac{1}{2}^+}{\sim}
\frac{(1-\mu^2)(e^{\frac{a\mu}{1-\mu^2}} - 1)}{\mu^2 a(2\epsilon - 1 )}.
\end{equation}
 The quantity $T(0,+)$ therefore becomes large in this limit. Moreover, the particle exits the system almost surely in this limit. The conditional average $\langle T(0)\rangle_{c,+}$ therefore equals $T(0,+)$, and becomes large in this limit.}} This is consistent with the vertical asymptote observed in the plot of $\langle  T(0)\rangle_{c,+}$ as a function of $\epsilon$ in Fig. \ref{figureCondTimeAlternatingDebug}.\\

$\bullet$ {\bf{Limit of short period ($a\ll 1$, at fixed $\epsilon>\frac{1}{2}$).}} Taylor expansion in Eq. (\ref{T0PlusAs}) yields
\begin{equation}\label{mueffas}
 T(0,+) \underset{a\ll 1}{\sim} \frac{2\alpha a( 1-2\epsilon) -\alpha a( 1-2\epsilon)}{-\mu \alpha a( 1-2\epsilon)} = \frac{1}{|\mu(1-2\epsilon)|}.
\end{equation}
 This is the value of the mean first-return time to the origin $T(0,+)$ for an RTP with constant negative drift $\mu(1-2\epsilon)$. As the particle exits the system almost surely, we obtain
\begin{equation}\label{C0PlusAs}
 \langle  T(0)\rangle_{c,+} \underset{a\ll 1, \epsilon > \frac{1}{2}}{\sim} \frac{1}{|\mu(1-2\epsilon)|}.
\end{equation}

\subsubsection{Non-zero survival probability ($J(a)>0$)}  

 Now assume $J(a)>0$, so that the survival probability is non-zero. The quantity $t(0)$ is given by Eq. (\ref{t0Expr}). It involves the quantity $e(0)$ as in Eq. (\ref{e0Sol}).  The expressions of $J(a)$ and $\Xi_-(a)$ were obtained in the case of alternating drifts in Eqs (\ref{JaVal},\ref{Xina}).
  The functions $\varphi$ and $\psi$ were defined in Eq. (\ref{phiPsiDef}) in terms of the field $F$ and the functions denoted by $e$ and $E$ (defined in Eqs (\ref{eExpr},\ref{EExpr})). Assembling these elements in the particular case of alternating drifts, 
we are led to the following expression in terms of the two parameters $\mu$ and $\epsilon$:
\begin{equation}\label{t0Comp}
\begin{split}
t(0)=&\frac{  e(0) \psi(a) -[ e(0) + 1] \varphi(a)   }{ 1 - e^{-J(a)}},\\
{\mathrm{with}}\quad e(0) =& -  \frac{\mu( 1-e^{-\alpha a(1-2\epsilon)})}{ \mu( 1-e^{-\alpha a(1-2\epsilon)}) + 1 - e^{-\alpha a} (2e^{ \alpha a \epsilon} -e^{2\alpha a\epsilon})},\\
 \varphi( a ) =& \frac{e^{-J(a)}}{1-\mu^2} \int_0^{a(1-\epsilon)} dy \left[  e(0)\mu -\left(E(0) +
\frac{e(0)}{\mu}\left( 1-e^{-\alpha y}  \right)\right)e^{\alpha y}
\right] \\     
 &\frac{e^{-J(a)}}{1-\mu^2} \int_{a(1-\epsilon)}^a dy
\left[   e(0)\mu +  \left(E(0) +
 \frac{e(0)}{\mu}   \left( -2 e^{-\alpha a(1-\epsilon)} + 1 + e^{-2\alpha a(1-\epsilon)}e^{+\alpha y}  \right)\right) e^{  \alpha( -y + 2a( 1 - \epsilon) )  }
\right],\\
\psi( a ) =& \frac{2}{1-\mu^2}\int_0^x \varphi(y) dy +
 \frac{\mu}{1-\mu^2}\left[
 \int_0^{a(1-\epsilon)} dy (E(y) - e(y) )
 - \int_{a(1-\epsilon)}^a dy  ( E(y) + e(y) )  
 \right],\\
e( x ) =&  e( 0 ) e^{-J(x)},\\
E(x) =& E( 0 ) + \frac{e(0)}{\mu}\left( 1-e^{-\alpha x} \right) \mathbbm{1}(x\in[0, a( 1-\epsilon)[)\\
&+\frac{e(0)}{\mu}\left( -2 e^{-\alpha a(1-\epsilon)} + 1 + e^{-2\alpha a(1-\epsilon)}e^{+\alpha x}   \right) \mathbbm{1}(x\in[ a( 1-\epsilon), a]),\\
E(0) =& e(0) + 1,\\
 J(x ) =& \alpha x \mathbbm{1}( x\in [0, a( 1-\epsilon)[)  +
 \alpha( -x + 2a( 1 - \epsilon) ) \mathbbm{1}( x\in [ a( 1-\epsilon), a[),\\
\varphi( x ) =& \frac{e^{-J(x)}}{1-\mu^2}\int_0^x dy
\left[
 e(y) \left( \mu \mathbbm{1}( y\in [ 0, a( 1-\epsilon)[ -\mu \mathbbm{1}( y\in [ a( 1-\epsilon), a[)  \right)  - E(y) \right]e^{+J(y)}, \quad(0\leq x \leq a).
\end{split}
\end{equation}
 The integrals in Eq. (\ref{t0Comp}) are amenable to explicit evaluation, but the explicit form as a function of $\epsilon$ is not particularly illuminating. {\textcolor{black}{The value of the conditional average of the exit time of a particle starting at the origin in a positive internal velocity state is 
 \begin{equation}\label{predForFig}
 \langle T(0)\rangle_{c,+} = \frac{T(0,+)}{E(0,+)} = \frac{2 t(0)}{E(0,+)}. 
 \end{equation}
  It has been evaluated (for $\epsilon$ in the interval $]0,1/2[$) based on Eq. (\ref{t0Comp}) using symbolic computation and is shown in Fig. \ref{figureCondTimeAlternatingDebug}}}, together with the results of numerical simulations. Once the conditional average of the first-return time to the origin is known, the spatial dependence of the conditional average follows from Eq. (\ref{resC2}). The plots of the conditional average $\langle T(y) \rangle_{c,\pm}$ as a function of $y$ on Fig. \ref{figureCxAlternating}, for $\epsilon = 0.25$ and $a=1$, illustrate the case of non-zero survival probability.\\

 Let us study the limits of the first-return time that we worked out in the case of almost-sure exit.\\

$\bullet$ {\bf{Limit of constant positive drift ($\epsilon\to 0$).}} The term of order zero in $\epsilon$ in the expression of $t(0)$ follows from 
\begin{equation}\label{t0Lim}
\begin{split}
e( 0 ) \underset{\epsilon \to 0}{\sim}&  - \frac{\mu}{1+\mu}, \\
J(x)  \underset{\epsilon \to 0}{\sim}&  \alpha x,\quad( 0\leq x \leq a),\\
t(0) \underset{\epsilon \to 0}{\sim}& \frac{1}{1-e^{-\alpha a}}\left[ 
-\frac{\mu}{1+\mu} \underset{\epsilon \to 0 }{\lim} \psi(a) +\frac{1}{1+\mu} \underset{\epsilon \to 0 }{\lim} \varphi(a)
\right].
\end{split}    
\end{equation}
The relevant limits of the integral terms are evaluated as
\begin{equation}
\begin{split}
E(x) \underset{\epsilon \to 0}{\sim}& \frac{e^{-\alpha x}}{1+\mu},\\
\underset{\epsilon \to 0 }{\lim} \varphi(x) =&\frac{e^{-\alpha x}}{1-\mu^2}
 \int_0^x \left[ -\frac{\mu^2}{1+\mu} -\frac{1}{1+\mu}  \right] dy= - \frac{1+\mu^2}{(1+\mu)(1-\mu^2)}x e^{-\alpha x},\quad( 0\leq x \leq a).\\
\end{split}
\end{equation}

\begin{equation}
\begin{split}
\underset{\epsilon \to 0 }{\lim} ~\psi(a) =& - \frac{2(1+\mu^2)}{(1+\mu)(1-\mu^2)^2} \int_0^a y e^{-\alpha y }dy +
\frac{2\mu}{(1+\mu)(1-\mu^2)} \int_0^a e^{-\alpha y} dy\\
=& \frac{1+\mu^2}{(1+\mu)( 1- \mu^2)}[ 
   y e^{-\alpha y}]_0^a - \frac{1+\mu^2}{\mu(1+\mu)( 1- \mu^2)} \int_0^a e^{-\alpha y} dy+
\frac{2\mu}{(1+\mu)(1-\mu^2)} \int_0^a e^{-\alpha y} dy\\
=& \frac{1+\mu^2}{\mu(1+\mu)( 1- \mu^2)} 
   a e^{-\alpha a} + \left( -\frac{1+\mu^2}{2\mu^2(1+\mu)} +
\frac{1}{(1+\mu)}\right)( 1 - e^{-\alpha a}).
\end{split}
\end{equation}
  Substituting into Eq. (\ref{t0Lim}) yields
\begin{equation}
\begin{split}
\underset{\epsilon \to 0 }{\lim}~t(0) = \frac{1-\mu}{2\mu(1+\mu)},
\end{split}    
\end{equation}
 which is the known value of $\frac{1}{2}T(0,+) = \frac{1}{2}\langle T(0)\rangle_{c,+} E(0,+)$ in the case of a constant positive drift (see Eqs (\ref{EConst},\ref{resIImu})).\\

$\bullet$ {\bf{Limit of small, positive $J(a)$ ($\epsilon\to \frac{1}{2}^-$).}} Let us write $\epsilon$ as
\begin{equation}
\epsilon:= \frac{1}{2} + h,\;\;\;\;\;\; h<0.
\end{equation}
Taylor expansion around $h=0$ yields
\begin{equation}\label{Jah}
\begin{split}
J(a)=& \frac{2\mu}{1-\mu^2}\int_0^{\frac{a}{2}(1-2h)}dy  - \frac{2\mu}{1-\mu^2}\int_{\frac{a}{2}(1-2h)}^a dy + o(h)\\
=& -\frac{2ah}{1-\mu^2} + o(h).
\end{split}
\end{equation}
 The quantity $J(a)$ is small and positive in this limit. The coefficient $e(0)$ is close to $0$ (its value in the cases where the particle leaves the system almost surely).
\begin{equation}
e(0) = o(1),\;\;\;\;\;E(0) = 1 + o(1 ),\;\;\;\;\;E(0,+) = 1 + o(1).
\end{equation}
 On the other hand, the quantity $\varphi(a)$ expressed in Eq. (\ref{t0Comp}) goes to a negative limit:
\begin{equation}
\varphi( a ) \underset{h\to 0}{\sim} -\frac{1}{\mu}( 1 - e^{\frac{-\alpha a}{2}}).
\end{equation}
 Hence, the equivalent of $t(0)$ (expressed in Eq. (\ref{t0Expr})):
\begin{equation}
 t(0)\underset{\epsilon\to \frac{1}{2}^-}{\sim} \frac{1-e^{-\alpha \frac{a}{2}}}{\mu J(a)} \sim\frac{(1-\mu^2)(1-e^{-\alpha \frac{a}{2}})}{a ( 1 - 2\epsilon)}.
\end{equation}
 The conditional average $\langle T(0) \rangle_{c,+}$ therefore becomes large when $\epsilon$ becomes close to $1/2$, which is reflected by the vertical asymptote in Fig. \ref{figureCondTimeAlternatingDebug}.\\

$\bullet$ {\bf{Limit of short period ($a\ll 1$, at fixed $\epsilon<\frac{1}{2}$).}} In this limit, we have noticed in Eq. (\ref{effectiveDrift}) that the exit probability is equal to the one of an RTP with the constant positive drift $\mu(1-2\epsilon)$, hence
\begin{equation}\label{equive}
e(0) \underset{a\ll 1}{\sim} -\frac{\mu(1-2\epsilon)}{1-\mu( 1 - 2\epsilon)},\quad\quad E(0) =e(0)+1\underset{a\ll 1}{\sim} 
\frac{1}{1-\mu( 1 - 2\epsilon)}.
\end{equation}
 Moreover, Taylor expansion at first order in $a$ yields
\begin{equation}\label{Jexpa}
 1-e^{-J(a)} = \alpha a (1-2\epsilon) + o(a).    
\end{equation}
The denominator of the expression of $t(0)$ in Eq. (\ref{t0Expr}) is therefore of order $a$. As the quantities $\varphi(a)$ and $\psi(a)$ are integrals (of continuous functions) on the short interval $[0,a]$, they are small in the limit of small $a$. 
 To obtain the limit of $t(0)$, we have to calculate the term of order $a$ in $\varphi(a)$ and $\psi(a)$. As $J(x)$ is small in this limit,\\
\begin{equation}
\begin{split}
 \varphi(x) =& \frac{e^{-J(x)}}{1-\mu^2}\int_0^x dy\left[
e(0) F(y) - E(0) e^{J(y)} + e(0) \Xi_-(y)e^{J(y)}
\right]\\
\underset{a\ll 1}{\sim}& \frac{1}{1-\mu^2} \int_0^x dy [ e(0) F(y) - e(0) - 1],\\
\end{split}
\end{equation}
hence
\begin{equation}
\varphi(a) \underset{a\ll 1}{\sim} \frac{e(0) [\mu(1-2\epsilon)-1] - 1}{1-\mu^2}a .
\end{equation}
Moreover,
\begin{equation}
\begin{split}
 \psi(a) =& 2 \int_0^a \frac{\varphi(y)}{1-F(y)^2} dy +  \int_0^a\frac{E(y) F(y) -e(y)}{1-F(y)^2}dy\\
 \underset{a\ll 1}{\sim}& \frac{1}{1-\mu^2}\int_0^a [E(0) F(y) -e(0)]dy\\
 \underset{a\ll 1}{\sim}& \frac{(e(0)+1) \mu(1-2\epsilon) - e(0)}{1-\mu^2}\int_0^a [E(0) F(y) -e(0)]dy.\\
\end{split}
\end{equation}
Combining with Eqs (\ref{equive},\ref{Jexpa}),
\begin{equation}
\begin{split}
 t(0) =& \frac{e(0)\psi(a) - (1+e(0))\varphi(a) }{1-e^{-J(a)}} \underset{a\ll 1}{\sim}  \frac{2e(0) + 1}{ \alpha( 1-\mu^2)(1-2\epsilon)}\\
 \underset{a\ll 1}{\sim}& \frac{1- \mu(1-2\epsilon)}{ 2\mu( 1 - 2\epsilon)(1+\mu(1-2\epsilon))}.
\end{split} 
\end{equation}
 We recognize the value of $t(0)$ for the RTP in a constant positive drift $\mu(1-2\epsilon)$, the effective drift  we identified in the limit of short period in the case of almost-sure exit (in Eq. (\ref{mueffas})). Hence
\begin{equation}
 \langle T(0)\rangle_{c,+} = \frac{T(0,+)}{E(0,+)} \underset{a\ll 1, \epsilon<\frac{1}{2}}{\sim} \frac{1}{\mu(1-2\epsilon)}.
\end{equation}
 Comparing with Eq. (\ref{C0PlusAs}), we observe that in both regimes (almost-sure exit and non-zero survival probability), the mean first-return time to the origin conditional on the exit coincides in the limit of a short period with the value obtained in the constant effective drift $\mueff$, with
\begin{equation}\label{mueffAll}
 \mueff = \mu( 1-2\epsilon).
\end{equation}

\begin{figure}
\includegraphics[width=1.1\textwidth,keepaspectratio]{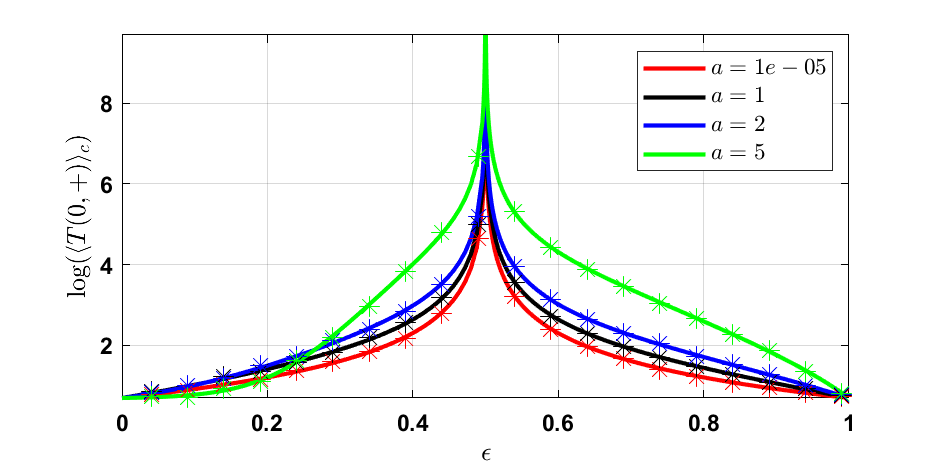}
\caption{The conditional average $\langle T(0)\rangle_{c,+}$ over trajectories that reach the origin,  for a particle starting at the origin in an internal positive velocity state, for a drift alternating between $\mu= 1/2$ on $[0,a(1-\epsilon)[$ and $-\mu$ on $[a(1-\epsilon), a[$, as a function of $\epsilon$. {\textcolor{black}{The solid line represents the theoretical predictions of Eq. (\ref{T0PlusSure}) for almost-sure exit ($J(a)<0$ and $\epsilon>1/2$), and of Eq. (\ref{predForFig}) for non-zero survival probability ($J(a)>0$ and $\epsilon<1/2$).  
The star symbols show results of numerical simulations (with the same trajectories as in Fig.  \ref{figureE0AlternatingDebug}).}}}
\label{figureCondTimeAlternatingDebug}
\end{figure}

\subsection{Generalization: first-return time to the origin in the limit of a short period}

  We have noticed that the conditional average of the 
 first-return time to the origin $\langle T(0)\rangle_{c,+}$ goes to a 
 well-defined  limit when $F$ alternates between opposite values on a short period interval. Let us generalize this situation by considering a periodic force field $F$ with a short period interval, meaning $a$ is short compared to the mean free path $v \gamma^{-1}$ of the RTP. In our units, this is the limit  $a\ll 1$. Assume $J(a)\neq 0$. In the limit of a short period interval, the integrals $\Xi_+(a)$ and $J(a)$ are both of order $a$.\\
 
 $\bullet$ {\bf{Almost-sure exit $(J(a)<0)$.}} From Eq. (\ref{t0Sure}), the mean first-return time to the origin is
 \begin{equation}
T(0,+) = -2 \frac{1}{e^{J(a)} - 1}\Xi_+(a)\underset{a \ll 1}{\sim} -\frac{2\Xi_+(a)}{J(a)}.
 \end{equation}
 As the exit probability $E(0,+)$ is equal to $1$, the conditional average of the first-return time to the origin becomes, in the limit of a short period interval, 
 \begin{equation}\label{T0CondAS}
\langle T(0)\rangle_{c,+} \underset{a\ll 1, J(a)<0}{\sim} \left|\frac{2\Xi_+(a)}{J(a)} \right|. 
 \end{equation}

 $\bullet$ {\bf{Non-zero survival probability $(J(a)>0)$.}}  In  the limit of a short period interval, we obtain the following limit of $E(0,+)$:
 \begin{equation}\label{E0PlusEq}
   E(0,+ ) = \frac{2\Xi_-(a) - ( 1 - e^{-J(a)})}{2\Xi_-(a) +  1 - e^{-J(a)}}
   \underset{a\ll 1}{\sim} \frac{1-\frac{J(a)}{2\Xi_-(a)}}{1+\frac{J(a)}{2\Xi_-(a)}}.  
 \end{equation}
 This value is identical to the exit probability of a particle starting at the origin with positive velocity, with a constant drift $\mueff$, with
 \begin{equation}\label{mueffDef}
\mueff := \underset{a\to 0}{\lim}\frac{J(a)}{2\Xi_-(a)}.
 \end{equation}
  This effective drift is evaluated as $\mu(1-2\epsilon)$ if $F$ consists of a rapidly alternating drift.\\

 On the other hand,
\begin{equation}
e(0) = -\frac{1-e^{-J(a)}}{1 - e^{-J(a)} + 2\Xi_-(a)} \underset{a \ll 1}{\sim}
-\frac{J(a)}{J(a) + 2\Xi_-(a)} = -\frac{\frac{J(a)}{2\Xi_-(a)}}{\frac{J(a)}{2\Xi_-(a)} + 1}.
\end{equation}
Substituting this equivalent of $e(0)$ into Eq. (\ref{phiPsiDef}) yields equivalents the quantity $\varphi(a)$:
\begin{equation}
\begin{split}
\varphi( a ) \underset{a \ll 1}{\sim}& e^{-J(a)} \int_0^a dy \left[ -\frac{\mueff}{1 + \mueff}\frac{F(y)}{1 - F(y)^2} -  \frac{1}{1 + \mueff} \frac{e^{J(y)}}{1 - F(y)^2} \right]\\
\underset{a \ll 1}{\sim}& -\frac{\mueff}{1 + \mueff}\frac{J(a)}{2} - \frac{1}{1 + \mueff} \Xi_+(a).
\end{split}
\end{equation}
As $\varphi(a)$ is of order $a$, and the function $\varphi$ is continuous, the term  first term in the expression of $\psi(a)$ is negligible compared to $a$, and
\begin{equation}
\begin{split}
\psi( a ) \underset{a \ll 1}{\sim}&
\int_0^x\frac{E(y) F(y) -e(y)}{1-F(y)^2}dy
\underset{a \ll 1}{\sim}
\frac{1}{1 + \mueff}\frac{J(a)}{2} + \frac{\mueff}{1 + \mueff} \Xi_-(a).
\end{split}
\end{equation}

Substituting the above equivalents into Eq. (\ref{t0Expr}) yields:
\begin{equation}
\begin{split}
T(0,+) =& 2t(0) = -\frac{\psi( a ) + \frac{2\Xi_-(a)}{1- e^{-J(a)}} \varphi(a)}{1-e^{-J(a)} +2\Xi_-(a)}\\
\underset{a \ll 1}{\sim}& \frac{-2}{J(a) + 2\Xi_-(a)} \left[  \psi(a ) +\frac{1}{\mueff} \varphi(a) \right] \\
\underset{a \ll 1}{\sim}& \frac{-2}{J(a) + 2\Xi_-(a)} \left[ \left( \frac{1}{1+\mueff} - \frac{1}{\mueff}\frac{\mueff}{1+\mueff}  \right) \frac{J(a)}{2} +   
 \left(\frac{\mueff}{ 1 + \mueff} -\frac{1}{\mueff}\frac{1}{1+\mueff} \right) \Xi_-(a) \right],\\
\end{split}
\end{equation}
where in the last step we used we used
\begin{equation}
    \Xi_-(a) \underset{a \ll 1}{\sim} \Xi_+(a)
     \underset{a \ll 1}{\sim} \int_0^a\frac{dy}{1 - F(y)^2}.
\end{equation}
Hence the expression of $T(0,+)$ coincides with its analogue for a constant drift $\mueff$: 
\begin{equation}
    T(0,+)\underset{a \ll 1}{\sim} 
    \frac{1-\mueff^2}{(1+\mueff)^2}\frac{2\Xi_-(a)}{J(a)} 
    \underset{a \ll 1}{\sim} \frac{1-\mueff}{1+\mueff}\frac{1}{\mueff}.
\end{equation}
 Dividing by the exit probability $E(0,+)$ and using Eq. (\ref{E0PlusEq}) yields
 \begin{equation}\label{eq2eff}
    \langle T(0)\rangle_{c,+}\underset{a \ll 1,J(a)>0}{\sim}\frac{1}{\mueff}.
\end{equation}

 Together with Eq. (\ref{T0CondAS}), the above limit implies that the first-return time averaged over trajectories that exit the system satisfies 
  \begin{equation}
 \langle T(0)\rangle_{c,+} \underset{a \ll 1,J(a)\neq 0}{\sim} \frac{1}{| \mueff |}.
  \end{equation}
This is the conditional average of the first-return time to the origin for a constant drift $\mueff$.\\
 
  We can recover the expression of the effective drift $\mueff$ using  a self-consistent argument (in the limit where the distance $a$ is short compared to the mean free path of the particle). In this limit, the length of a typical run is large compared to $a$. When the internal velocity state is $\pm 1$, this length can be approximated by $N_\pm a$ (where $N_\pm$ is a large integer). Let us look for an effective drift $\mu'_{\mathrm{eff}}$ such that the length of a typical run (with internal velocity state $\pm 1$) is $( \mu'_{\mathrm{eff}}\pm 1)t_\pm$ (where $t_\pm$ is the duration of a typical run, which depends on the internal velocity state of the particle). Integrating the equations of motion (Eqs (\ref{eqn:t+},\ref{eqn:t-})), we can approximate the duration $t_\pm$ in each of the two internal velocity states by
\begin{equation}
  \begin{split}
 t_\pm \simeq N_\pm a (\mu'_{\mathrm{eff}}  \pm 1)^{-1} \simeq \int_0^{N_{\pm} a} \frac{dx}{F(x) \pm 1}.  
  \end{split}
\end{equation}
 Dividing the above expression by $N_+$ (resp. $N_-$) and using periodicity yields:
\begin{equation}
 a(\mu'_{\mathrm{eff}} \pm 1)^{-1} = \int_0^{ a} \frac{dx}{F(x) \pm 1}.  
\end{equation}
The sum and difference of the above two equations yield
\begin{equation}
 \begin{split}
 \frac{2a \mu'_{\mathrm{eff}}}{1 - {\mu'_{\mathrm{eff}}}^2}=& \int_0^{a} \frac{2F( x ) dx}{1 - F(x)^2},\\  
 \frac{2a }{ 1- {\mu'_{\mathrm{eff}}}^2}  =&  \int_0^{a} \frac{2 dx}{1 - F(x)^2}.  
 \end{split}
\end{equation}
 {\textcolor{black}{The r.h.s. of the  first equation is the quantity $J(a)$. Moreover,
 the effective drift $\mu'_{\mathrm{eff}}$ follows as
 \begin{equation}
 \mu'_{\mathrm{eff}} = \frac{  \int_0^{a} \frac{2F( x ) dx}{1 - F(x)^2}}{\int_0^a \frac{2 dx}{1 - F(x)^2}},   
 \end{equation}
 which has the same sign as the quantity $J(a)$.}} Evaluating the integrals in the above equations in the case of an alternating force field yields the expression of the effective drift $\mueff$ obtained from direct calculations (in the limit of a short period interval $a$) in Eq. (\ref{mueffAll}). More generally,
\begin{equation}
\frac{  \int_0^{a} \frac{2F( x ) dx}{1 - F(x)^2}}{\int_0^a \frac{2 dx}{1 - F(x)^2}}\underset{a \ll 1}{\sim} \frac{J(a)}{2\Xi_-(a)}.
\end{equation}
 As we are working is the limit where $a$ is short, 
 the effective drift $\mu'_{\mathrm{eff}}$ is identical to the quantity $\mueff$ introduced in Eq. (\ref{mueffDef}).\\

{\textcolor{black}{The quantity $J(a)$, which emerged from the integration of the evolution equations of the exit probability, and from the above self-consistent argument is the case of a short period, may be interpreted in terms of the active external potential $W$ defined in \cite{le2020velocity}. In our system of units\footnote{For the original expression of the active external potential in terms of the parameters $v$ and $\gamma$ and an explicit change of variables, see Appendix \ref{appUnits}.}, $W$ is expressed in terms of an arbitrary position $x_0$ as 
\begin{equation}\label{WExpr}
 W( x ) := -2 \int_{x_0}^x dy \frac{F(y)}{1 - F(y)^2},
\end{equation}
so that
\begin{equation}
J(a) = W( 0 ) - W( a ).
\end{equation}
}}

{\textcolor{black}{
 In \cite{le2020velocity}, an RTP is studied on the entire real line in the presence of a periodic force. If this force takes only subcritical values and the particle can access every position on the real line, there is no stationary solution for the probability of presence  $P(x,t)$ of the particle at position $x$ and time $t$, but there is one 
 for the periodised version $\tilde{P}(x,t):=\sum_{n\in\mathbb{Z}} P(x+na,t)$, which gives rise to a stationary current. This stationary current was shown (in Section A of Supplementary Material  
 of \cite{le2020velocity}) to have the same sign as $W(0)-W(a)$. In our case there is no stationary solution either, but we recognize $J(a)$ as the quantity  governing the effective velocity of an RTP in a periodic force field. If $J(a)>0$ the active external potential drives the particle away from the origin, and if $J(a)<0$ it drives it towards the origin.}}

\section{Discussion and outlook}

We have obtained a closed-form expression for the exit probability and mean first-passage time at the origin of a run-and-tumble particle in any spatially-periodic force field whose values allow the particle to access any position on the positive half-line. The particle exits the system in finite time almost surely if and only if a certain integral of a function of the force field (denoted by $J(a)$ and defined in Eq. (\ref{defJ})) {\textcolor{black}{is negative or zero.}} This condition has been established by a Markov-chain argument. Using periodicity to fix integration constants, we have calculated the average of the first-passage time at the origin, over the trajectories that exit the system in finite time. The results generalize those obtained in the presence of a constant subcritical drift \cite{de2021survival}.  In particular, the mean exit time of a particle starting at the origin with a positive velocity becomes large at the separation between the two regimes of almost-sure exit and non-zero survival probability, just as in the case of a constant drift. We have evaluated the exit probability and conditional average of the first-passage time to the origin explicitly over trajectories that exit the system.\\

The dependence on the initial distance to the origin (at integer multiples of the spatial period) is an affine function of the initial position. Direct calculations become more tedious in the case of non-zero survival probability, but a renewal argument based on periodicity yields the spatial dependence of the mean first-passage time, once the differential equations have been integrated on one period interval.
We have noticed that the exit probability and the first-return time $T(0,+)$ takes the same form as if $F$ consisted of an effective constant drift.\\

 The model can be generalized in several ways. As noted in \cite{gueneau2024run}, in the case of an almost-sure exit, the higher moments of the exit time can in principle be calculated iteratively, based on the Taylor expansion of the Laplace transform of the flow of particles through the origin. To make the model more realistic, one could take into account the nonzero duration of tumbles, and model the tumble time as an exponential random variable (in colonies of bacteria, tumble times and run times have the same order of magnitude \cite{berg2004coli}). Such a model has been solved in dimension $d$ for the survival probability of a run-and-tumble particle with an absorbing hyperplane and zero force field, if the particle starts on the hyperplane \cite{mori2020universal}. Quite remarkably, the result does not depend on the dimension. The model of an RTP in higher dimension with an absorbing hyperplane and a constant drift has been solved in the presence of a constant drift in \cite{de2021survival}. Technically, the evolution equation of the system becomes a coupled system of  backward Fokker--Planck  equations (one for every value of the projection of the internal velocity on the normal direction to the absorbing hyperplane).
 More recent developments have included the effect of inertia in the model of the run-and-tumble particle in the background of a harmonic potential \cite{dutta2024harmonically}. More generally, inertia induces an additional time scale,  and distinct dynamical regimes have been identified depending on how this time scale compares to the mean flipping time $\gamma^{-1}$ and to the current time \cite{dutta2024inertial}. It would be interesting to estimate corrections to the present result in the imit where the inertial time scale is short compared to the mean flipping time.



\begin{appendices}

\section{Systems of units}\label{appUnits}
{\textcolor{black}{
 Let us demonstrate on an example how to relate the expressions obtained in this paper (where $v$ and $\gamma$ are set to 1) to expressions containing the parameters $v$ and $\gamma$.}}\\ 
 
 {\textcolor{black}{The original definition (in \cite{le2020velocity}) of the active external potential at position $X$ (call it $W_{\mathrm{ori}}(X)$)  contains the parameters $v$ and $\gamma$. The quantity is given in terms of an arbitrary reference position $X_0$ (serving as a gauge choice), for a field $f$ depending on time, as
\begin{equation}
 W_{\mathrm{ori}}( X ) := -2 \gamma\int_{X_0}^X dY \frac{f(Y)}{v^2 - f(Y)^2}.
\end{equation}
 Using $v$ as the unit of velocity and $\gamma$ as the unit of frequency, the unit of length becomes $v\gamma^{-1}$, and 
 the positions $X_0$ and $X$ become the  coordinates $x$ and $x_0$. With the following expressions of the dimensionless quantities $x,x_0,y,F$:
 \begin{equation}
 X =  v\gamma^{-1} x,\qquad X_0 =  v\gamma^{-1} x_0,\qquad Y =  v\gamma^{-1} y, \qquad f(Y) = v F( y ),
 \end{equation}
we can re-express the value of the active external potential by changing the integration variable to $y$:
\begin{equation}
 W_{\mathrm{ori}}( X ) =  -2\gamma \int_{x_0}^x dy v\gamma^{-1} \frac{v F(y)}{v^2-v^2F(y)^2}= -2\int_{x_0}^x dy \frac{ F(y)}{1-F(y)^2}.
 \end{equation}
The value $W(x):=W_{\mathrm{ori}}(X)$ is the one reported in Eq. (\ref{WExpr}). 
}}

\section{Evaluation of the functions $J$ and $\Xi_\pm$ using periodicity}\label{appInt}
 Let $F$ be a periodic function of period $a$. The functions
$J$, $\Xi_-$ and $\Xi_+$ defined in Eqs (\ref{defJ},\ref{defXi}) can be expressed on the entire positive half-line in terms of their values of $F$ in the interval $[0,a]$. Indeed, for $x$ in $[0,a[$ and a positive integer $N$, the periodicity of $F$ yields:   
\begin{equation}\label{JPer}
\begin{split}
 J( x + Na ) =& \sum_{k=0}^{N-1} \int_{ka}^{(k+1)a} dy \frac{2F(y)}{1-F(y)^2} + \int_{Na}^x dy \frac{2F(y)}{1-F(y)^2}\\
=& N J(a ) + J(x), 
\end{split}
\end{equation} 

Assuming $J(a) \neq 0$, the same reasoning yields
\begin{equation}\label{XiPer}
\begin{split}
\Xi_n( x + Na ) =& \left(\sum_{k=0}^{N-1} \int_{ka}^{(k+1)a} \frac{e^{nJ(x)}}{1-F(x)^2}dx\right) +  \int_{Na}^{Na + x} \frac{e^{nJ(v)}}{1-F(v)^2}dv\\
=& \left( \sum_{k=0}^{N-1} e^{nkJ(a)} \Xi_n( a) \right) +  e^{nNJ(a)}\Xi_n(x)\\
=&\frac{1- e^{nNJ(a)}}{1 - e^{nJ(a)}}\Xi_n( a) + e^{nNJ(a)}\Xi_n(x),\\
&\qquad x\in [0,a],\qquad N\in\mathbb{N},\qquad n\in \{-,+\}.
\end{split}
\end{equation}

An integral of the following form appears in the expression of the mean first-passage time in Eq. (\ref{Txpm}):
\begin{equation}
\begin{split}
  \int_0^{x+Na} \Xi'_m(v) \Xi_n(v ) dv =& 
\sum_{k=0}^{N-1}\int_{ka}^{(k+1)a} \Xi'_m(v) \Xi_n(v ) dv + \int_{Na}^{x+Na} \Xi'_m(v) \Xi_n(v ) dv\\
=& \sum_{k=0}^{N-1}\int_{0}^{a} \Xi'_m(ka + v) \Xi_n(ka + v ) dv + \int_{0}^{x} \Xi'_m(Na + v) \Xi_n(Na + v ) dv.
\end{split}
\end{equation}
 Substituting the expression of $\Xi_n(Na+v)$ obtained in Eq. (\ref{XiPer}) (and its derivative w.r.t. $v$, which reads $\Xi'_m( ka + v ) = e^{mkJ(a)}\Xi'_m( v )$) yields
\begin{equation}\label{intXiPrimeXimn}
\begin{split}
  \int_0^{x+Na} \Xi'_m(v) \Xi_n(v ) dv =& 
 \sum_{k=0}^{N-1} e^{mkJ(a)}\int_{0}^{a} \Xi'_m( v) \Xi_n(ka + v ) dv + e^{mNJ(a)}\int_{0}^{x} \Xi'_m( v) \Xi_n(Na + v ) dv\\
  =& \sum_{k=0}^{N-1} e^{mkJ(a)}\int_{0}^{a} \Xi'_m( v) 
  \left[\frac{1- e^{nkJ(a)}}{1 - e^{nJ(a)}}\Xi_n( a) + e^{nkJ(a)}\Xi_m(v)\right]
  dv\\
  &+ e^{mNJ(a)}\int_{0}^{x} \Xi'_m( v) 
  \left[   
  \frac{1- e^{nNJ(a)}}{1 - e^{nJ(a)}}\Xi_n( a) + e^{nNJ(a)}\Xi_n(v)
  \right]
  dv \\
=&   \Xi_m(a) \Xi_n(a) \sum_{k=0}^{N-1} \frac{e^{mkJ(a)}(1 - e^{nkJ(a)}) }{1-e^{nJ(a)}} + \Xi_m(x)\Xi_n(a) e^{mNJ(a)} \frac{1- e^{nNJ(a)}}{1- e^{nJ(a)}}\\
+& \sum_{k=0}^{N-1} e^{(m+n) k J(a)} \int_0^a \Xi'_m(v)  \Xi_n(v) dv + e^{(m+n)NJ(a)} \int_0^x \Xi'_m(v) \Xi_n(v) dv.\\
\end{split}
\end{equation}
In particular,
\begin{equation}\label{intXiPrimeXiMP}
\begin{split}
  \int_0^{x+Na} \Xi'_-(v) \Xi_+(v ) dv  =&   \frac{\Xi_-(a) \Xi_+(a)}{ 1 - e^{J(a)}} \left(   -N + \frac{1-e^{-NJ(a)}}{1-e^{-J(a)}}  \right)\\   
 +& \Xi_-(x)\Xi_+(a) e^{-NJ(a)} \frac{1- e^{NJ(a)}}{1- e^{J(a)}}\\
+& N \int_0^a \Xi'_-(v)  \Xi_+(v) dv+ \int_0^x \Xi'_-(v) \Xi_+(v) dv,\\
=&  N\left[  \int_0^a \Xi'_-(v)  \Xi_+(v) dv  - \frac{\Xi_-(a)\Xi_+(a)}{1-e^{J(a)}} \right]\\
&+ \frac{\Xi_+(a)}{1-e^{J(a)}}\left( \frac{\Xi_-(a)}{1- e^{-J(a)}} +\Xi_-(x)    \right) + \int_0^x\Xi_-(v)\Xi_+(v) dv\\
&- \frac{\Xi_+(a)e^{-NJ(a)}}{1-e^{J(a)}}\left( \frac{\Xi_-(a)}{1- e^{-J(a)}} +\Xi_-(x)    \right),
\end{split}
\end{equation}
where we have collected the terms of order $N$, $1$ and $e^{-NJ(a)}$.

\section{Evaluation of the functions $\varphi$ and $\psi$ using periodicity}\label{periodApp}
\subsection{Expression of $\varphi(x)$ in the $N$-th period interval}
 Consider a nonnegative integer $N$, and $x$ in $[0,a[$. To express the function $t$ in the interval $[Na, (N+1)a]$ in the case of non-zero survival probability ($J(a)>0$), we need to evaluate the function $\varphi$ defined in Eq. (\ref{phiPsiDef}) in this interval:
\begin{equation}
\varphi( Na + x ) = e^{-J(x+Na)}\left(
\sum_{k=0}^{N-1}
\int_{ka}^{(k+1)a} dy \frac{e(y)F(y) - E(y)}{1-F(y)^2}e^{J(y)}  
+
\int_{Na}^{Na + x} dy  \frac{e(y)F(y) - E(y)}{1-F(y)^2} e^{J(y)}.
\right)
\end{equation}
 Let us assume $J(a)>0$. The survival probability is non-zero, hence  
$e(y) = e(0) \exp( -J(y) )$ and $E( y ) = E(0) + 2e(0) \Xi_-(y)$, with  $e(0)<0$ given by Eq. (\ref{e0Sol}). As $F$ is $a$-periodic, 
 \begin{equation}
\begin{split}
\varphi&( Na + x ) =\\
&e^{-NJ(a) - J(x)}\left[
\sum_{k=0}^{N-1}
\int_{0}^{a} dy \left( \frac{F(y)}{1- F(y)^2} - \frac{ E(0) + 2e(0)\left[  \frac{1- e^{-kJ(a)}}{1 - e^{-J(a)}}\Xi_-( a) + e^{-kJ(a)}\Xi_-(y)  \right]}{1-F(y)^2} e^{kJ(a) + J(y)} \right) \right.\\
& \left.+
\int_{0}^{ x} dy \left(
\frac{F(y)}{1- F(y)^2} - \frac{ E(0) + 2e(0)\left[  \frac{1- e^{-NJ(a)}}{1 - e^{-J(a)}}\Xi_-( a) + e^{-NJ(a)}\Xi_-(y)  \right]}{1-F(y)^2} e^{NJ(a) + J(y)} 
\right)\right]\\
=& e^{-NJ(a) - J(x)}\times N\left[ \int_0^a dy \frac{F(y)}{1-F(y)^2}
+\frac{2e(0)\Xi_-(a)}{1-e^{-J(a)}} \int_0^a dy\frac{e^{J(y)}}{1-F(y)^2} -2e(0)
 \int_0^a dy\frac{e^{J(y)}\Xi_-(y)}{1-F(y)^2}
  \right]\\
&-e^{-NJ(a) - J(x)}\times \int_0^a dy \left( E(0) +\frac{2e(0)}{1-e^{-J(a)}}\Xi_-(a) \right) \frac{1-e^{N J(a)}}{(1-F(y)^2)( 1 - e^{J(a)})}e^{J(y)}\\
&+ e^{-NJ(a) - J(x)}  
\left[
\int_0^x dy \frac{F(y)}{1-F(y)^2} -\left( E(0) + 2e(0) \frac{1-e^{-NJ(a)}}{1-e^{-J(a)}} \Xi_-(a)\right)e^{NJ(a)} \int_0^x dy\frac{e^{J(y)}}{1-F(y)^2}\right.\\
&\left. -2e(0) \int_0^x dy\frac{e^{J(y)}\Xi_-(y)}{1-F(y)^2}
\right].
\end{split}     
 \end{equation}
From the value of $e(0)$ obtained in Eq. (\ref{e0Sol}) in the case $J(a)>0$, and the boundary condition $E(0)=e(0) + 1$, the coefficient of $N$ vanishes in the above expression. Indeed, using the expression of $e(0)$ in Eq. (\ref{e0Sol}),
\begin{equation}\label{melo}
 \begin{split}
  E(0) + \frac{2e(0)}{1-e^{-J(a)}}\Xi_-(a) =&
  e(0)\frac{1-e^{-J(a)} +2\Xi_-(a) }{1 - e^{-J(a)}}    + 1\\
  =&-1 + 1 = 0.
\end{split}
\end{equation}
Hence the simplification
\begin{equation}
\begin{split}
\varphi&( Na + x ) =\\
&e^{-NJ(a) - J(x)}\left[
\sum_{k=0}^{N-1}
\int_{0}^{a} dy \left( \frac{e(0)F(y)}{1- F(y)^2} - \frac{ E(0) + 2e(0)\left[  \frac{1- e^{-kJ(a)}}{1 - e^{-J(a)}}\Xi_-( a) + e^{-kJ(a)}\Xi_-(y)  \right]}{1-F(y)^2} e^{kJ(a) + J(y)} \right) \right.\\
& \left.+
\int_{0}^{ x} dy \left(
\frac{e(0)F(y)}{1- F(y)^2} - \frac{ E(0) + 2e(0)\left[  \frac{1- e^{-NJ(a)}}{1 - e^{-J(a)}}\Xi_-( a) + e^{-NJ(a)}\Xi_-(y)  \right]}{1-F(y)^2} e^{NJ(a) + J(y)} 
\right)\right]\\
=& e^{-NJ(a) - J(x)}\times N\left[ \int_0^a dy \frac{e(0)F(y)}{1-F(y)^2}
+\frac{2e(0)\Xi_-(a)}{1-e^{-J(a)}} \int_0^a dy\frac{e^{J(y)}}{1-F(y)^2} -2e(0)
 \int_0^a dy\frac{e^{J(y)}\Xi_-(y)}{1-F(y)^2}
  \right]\\
&+ e^{-NJ(a) - J(x)}  
\left[
\int_0^x dy \frac{e(0)F(y)}{1-F(y)^2}  + \frac{2e(0)\Xi_-(a)}{1-e^{-J(a)}} \int_0^x dy\frac{e^{J(y)}}{1-F(y)^2} -2e(0) \int_0^x dy\frac{e^{J(y)}\Xi_-(y)}{1-F(y)^2}
\right]\\
=&  e^{-NJ(a) - J(x)}\times N\left[ \int_0^a dy \frac{e(0)F(y)}{1-F(y)^2}
+\frac{2e(0)\Xi_-(a)\Xi_+(a)}{1-e^{-J(a)}} -2e(0)
 \int_0^a dy\Xi'_+(y)\Xi_-(y)
  \right]\\
&+ e^{-NJ(a) - J(x)}  
\left[
\int_0^x dy \frac{e(0)F(y)}{1-F(y)^2}  + \frac{2e(0)\Xi_-(a)\Xi_+(x)}{1-e^{-J(a)}}  -2e(0) \int_0^x dy\Xi'_+(y)\Xi_-(y)
\right]\\
=& Ne^{-NJ(a)} e^{- J(x)} \left[ \int_0^a dy \frac{e(0)F(y)}{1-F(y)^2}
-\frac{2\Xi_-(a)\Xi_+(a)}{1-e^{-J(a)} + 2\Xi_-(a)} -2e(0)
 \int_0^a dy\Xi'_+(y)\Xi_-(y)
  \right]\\
&+ e^{-NJ(a) - J(x)}  
\left[
\int_0^x dy \frac{e(0)F(y)}{1-F(y)^2}  -\frac{2\Xi_-(a)\Xi_+(x)}{1-e^{-J(a)} + 2\Xi_-(a)}  -2e(0) \int_0^x dy\Xi'_+(y)\Xi_-(y)
\right].\\
\end{split}
\end{equation}
The expression therefore consists of terms of order $N e^{-N(J(a)}$ and terms of order $e^{-N(a)}$:\\
\begin{equation}\label{phiExprNx}
\varphi(Na +x)= C e^{-J(x)} Ne^{-NJ(a)} + D(x) e^{-J(x)} e^{-NJ(a)},
\end{equation}
with the notations
\begin{equation}\label{CDDef}
\begin{split}
C:=& \int_0^a dy \frac{e(0)F(y)}{1-F(y)^2}
-\frac{2\Xi_-(a)\Xi_+(a)}{1-e^{-J(a)} + 2\Xi_-(a)} -2e(0)
 \int_0^a dy\Xi'_+(y)\Xi_-(y),\\
 D(x)& := \int_0^x dy \frac{e(0)F(y)}{1-F(y)^2}  -\frac{2\Xi_-(a)\Xi_+(x)}{1-e^{-J(a)} + 2\Xi_-(a)}  -2e(0) \int_0^x dy\Xi'_+(y)\Xi_-(y).
\end{split}
\end{equation}

\subsection{Expression of $\psi$ in the $N$-th period interval}
Consider again a nonnegative integer $N$, and $x$ in the interval $[0,a[$. Using the definition in Eq. (\ref{phiPsiDef}),  
\begin{equation}
\psi(x+Na) = 2\int_0^{x+Na} \frac{\varphi(y)}{1-F(y)^2} dy + \int_0^{x+Na} \frac{E(y) F(y) -e(y)}{1-F(y)^2}dy.
\end{equation}

To calculate the first term, we need to evaluate the following integral:
\begin{equation}\label{sumIntPhi}
\begin{split}
\int_0^{x+Na}dy\frac{\varphi(y)}{1-F(y)^2}=& \sum_{k=0}^{N-1} \int_0^a dy  \frac{\varphi(ka +y)}{1-F(y)^2} + \int_0^x dy  \frac{\varphi(Na +y)}{1-F(y)^2}\\
=& \sum_{k=0}^{N-1}  \mathcal{I}_k(a) + \mathcal{I}_N(x),\\
\end{split}
\end{equation}
with $\mathcal{I}_N(x)$ defined for any nonnegative integer $N$ and any $x$ in $[0,a]$ as:
\begin{equation}
\begin{split}
\mathcal{I}_N(x):=& \int_0^x dy  \frac{\varphi(Na +y)}{1-F(y)^2}.\\
\end{split}
\end{equation}
Using Eq. (\ref{phiExprNx}) yields
\begin{equation}\label{INExpr}
\begin{split}
\mathcal{I}_N(x) =& \int_0^x dy  \frac{1}{1-F(y)^2}\left[
C e^{-J(y)} Ne^{-NJ(a)} + D(y) e^{-J(y)} e^{-NJ(a)}\right]\\
=&  C  \Xi_-(x )Ne^{-NJ(a)} + \left(   \int_0^x dy\frac{D(y) e^{-J(y)}}{1-F(y)^2}\right) e^{-NJ(a)}\\
=&  C  \Xi_-(x )Ne^{-NJ(a)} + K(x) e^{-NJ(a)},
\end{split}
\end{equation}
where we used the notations introduced in Eq. (\ref{CDDef}), and the notation 
\begin{equation}\label{KDef}
K(x) = \int_0^x dy \frac{D(y) e^{-J(y)}}{1- F(y)^2},
\end{equation}

 The sums in Eq. (\ref{sumIntPhi}) are  evaluated using the identity
\begin{equation}
\begin{split}
 \sum_{k=0}^{N-1} k x^k &= x \frac{d}{dx}\left( \sum_{k=0}^{N-1} x^k 
 \right) = x\frac{d}{dx}\left( \frac{ 1- x^{N}}{ 1- x} \right)\\
 =& \frac{ x(1- x^{N})}{ (1- x)^2} -N \frac{x^{N}}{1 - x} = \frac{(N-1)x^{N+1} - N x^N +x}{(1-x)^2}\\
 =& \frac{x}{(1-x)^2} -\frac{1}{ 1- x}Nx^N  -\frac{x}{(1-x)^2}x^N.
 \end{split}
\end{equation}

\begin{equation}\label{termpsi1}
\begin{split}
\int_0^{x+Na}dy\frac{\varphi(y)}{1-F(y)^2}=& 
 C\Xi_-(a)\frac{e^{-J(a)}}{(1-e^{-J(a)})^2} + \frac{K(a)}{1-e^{-J(a)}}\\
 &+ C\left( \Xi_-(x) - \frac{\Xi_-(a)}{1-e^{-J(a)}} \right)N e^{-NJ(a)}\\
 &+ \left(  K(x) - \frac{K(a)}{1-e^{-J(a)}}- \frac{C\Xi_-(a) e^{-J(a)}}{(1-e^{-J(a)})^2}  \right) e^{-NJ(a)}.
\end{split}
\end{equation}

The second term in the expression of $\psi(x+Na)$ is evaluated by taking similar calculational steps:
\begin{equation}
\begin{split}
\int_0^{Na+x} dy  &\frac{E(y)F(y) - e(y)}{1-F(y)^2} =\\ 
&\sum_{k=0}^{N-1}
\int_{ka}^{(k+1)a} dy \frac{E(y)F(y) - e(y)}{1-F(y)^2} 
+
\int_{Na}^{Na + x} dy  \frac{E(y)F(y) - e(y)}{1-F(y)^2}\\
=& \sum_{k=0}^{N-1}
\int_{0}^{a} dy \frac{[E(0) + 2e(0) \Xi_-(ka + y ) ]F(y) - e(0) \exp( -J(ka  +y ) )}{1-F(y)^2} \\
&+
\int_{0}^{ x} dy  \frac{[E(0) + 2e(0) \Xi_-(Na + y ) ] F(y) - e( 0) \exp( -J(Na  +y ) ) }{1-F(y)^2}\\
=& \sum_{k=0}^{N-1}
\int_{0}^{a} dy \frac{[E(0) + 2e(0) \left( \frac{1- e^{-kJ(a)}}{1 - e^{-J(a)}}\Xi_-( a) + e^{-kJ(a)}\Xi_-(y) \right)  ]F(y) - e(0) e^{-kJ(a)} e^{-J(y)}}{1-F(y)^2} \\
&+
\int_{0}^{ x} dy  \frac{[E(0) + 2e(0) \left( \frac{1- e^{-NJ(a)}}{1 - e^{-J(a)}}\Xi_-( a) + e^{-NJ(a)}\Xi_-(y) \right)  ] F(y) - e( 0)  e^{-NJ(a)}  e^{-J(y)} }{1-F(y)^2}\\
=& \sum_{k=0}^{N-1}\int_0^a dy \left(  E(0) + \frac{2e(0)}{1-e^{-J(a)}}\Xi_-(a) \right) \frac{F(y)}{1-F(y)^2}\\
&+ \sum_{k=0}^{N-1} e^{-kJ(a)} e(0)\int_0^a dy \left[ \left(  -\frac{2}{1-e^{-J(a)}}\Xi_-(a) +2\Xi_-(y) \right)  \frac{F(y)}{ 1-F(y)^2}
-\frac{e^{-J(y)}}{ 1-F(y)^2}\right]\\
&+\int_0^x dy \left(  E(0) + \frac{2e(0)}{1-e^{-J(a)}}\Xi_-(a) \right) \frac{F(y)}{1-F(y)^2}\\
&+ e^{-NJ(a)} e(0)\int_0^x dy \left[ \left(  -\frac{2}{1-e^{-J(a)}}\Xi_-(a) +2\Xi_-(y) \right)  \frac{F(y)}{ 1-F(y)^2}
-\frac{e^{-J(y)}}{ 1-F(y)^2}\right].\\
\end{split}
\end{equation}
Using again Eq. (\ref{melo}), and grouping the remaining terms according to their large-$N$ behavior yields
\begin{equation}\label{groupTermPsi}
\begin{split}
\int_0^{Na+x} dy  &\frac{E(y)F(y) - e(y)}{1-F(y)^2} =\\ 
& \sum_{k=0}^{N-1} e^{-kJ(a)} e(0)\int_0^a dy \left[ \left(  -\frac{2}{1-e^{-J(a)}}\Xi_-(a) +2\Xi_-(y) \right)  \frac{F(y)}{ 1-F(y)^2}
-\frac{e^{-J(y)}}{ 1-F(y)^2}\right]\\
&+ e^{-NJ(a)} e(0)\int_0^x dy \left[ \left(  -\frac{2}{1-e^{-J(a)}}\Xi_-(a) +2\Xi_-(y) \right)  \frac{F(y)}{ 1-F(y)^2}
-\frac{e^{-J(y)}}{ 1-F(y)^2}\right]\\
=& + \sum_{k=0}^{N-1} e^{-kJ(a)} G(a) + e^{-NJ(a)} G(x),
\end{split}
\end{equation}
with the notation 
\begin{equation}
 G(x) := e(0)\int_0^x dy \left[ \left(  -\frac{2}{1-e^{-J(a)}}\Xi_-(a) +2\Xi_-(y) \right)  \frac{F(y)}{ 1-F(y)^2}
-\frac{e^{-J(y)}}{ 1-F(y)^2}\right].
\end{equation}

Working out the geometric sum and grouping the terms according to their dependence on $N$ yields:
\begin{equation}\label{termpsi2}
\int_0^{Na+x} dy  \frac{E(y)F(y) - e(y)}{1-F(y)^2} =  \frac{G(a)}{1-e^{-J(a)}} + e^{-NJ(a)}\left(
G(x ) - \frac{G(a)}{1-e^{-J(a)}} \right).
\end{equation}
With these notations:
\begin{equation}
    \psi( a ) = 2K(a ) + G(a ).
\end{equation}
Adding up the two terms in the expression of $\psi(x+Na)$ (obtained in Eq. (\ref{termpsi1},\ref{termpsi2})) yields
\begin{equation}
\begin{split}
\psi( x + Na ) =& 2C\Xi_-(a)\frac{e^{-J(a)}}{(1-e^{-J(a)})^2} + 2\frac{K(a)}{1-e^{-J(a)}} + \frac{G(a)}{1-e^{-J(a)}}\\
 &+ 2C\left( \Xi_-(x) - \frac{\Xi_-(a)}{1-e^{-J(a)}} \right)N e^{-NJ(a)}\\
 &+ \left(  2K(x) - 2\frac{K(a)}{1-e^{-J(a)}}- 2\frac{C\Xi_-(a) e^{-J(a)}}{(1-e^{-J(a)})^2}  +G(x ) - \frac{G(a)}{1-e^{-J(a)}}\right) e^{-NJ(a)}.   
\end{split}
\end{equation}

Combining yields the explicit expressions of $t(x+Na)$ and $T(x+Na)$ as:
\begin{equation}
\begin{split}
t(x+Na ) =& t(0) e^{-J(Na + x)} + \varphi(x+Na)\\
=& C e^{-J(x)} Ne^{-NJ(a)} + ( t(0) + D(x) )e^{-J(x)} e^{-NJ(a)},\\
T(x+Na) =& t(0)( 1 + 2\Xi_-(Na + x ) ) +\psi(x+Na)\\
=&  t(0)\left[  1 + \frac{2\Xi_-(a)}{1- e^{-J(a)}}   \right] +
2C\Xi_-(a)\frac{e^{-J(a)}}{(1-e^{-J(a)})^2} + 2\frac{K(a)}{1-e^{-J(a)}} + \frac{G(a)}{1-e^{-J(a)}}\\
 &+ 2C\left( \Xi_-(x) - \frac{\Xi_-(a)}{1-e^{-J(a)}} \right)N e^{-NJ(a)}\\
 &+ \left[ 2t(0)\left( -\frac{\Xi_-(a)}{1-e^{-J(a)}} +\Xi_-(x)\right) - 2\frac{\varphi(a)\Xi_-(a)}{(1-e^{-J(a)})^2}- \frac{\psi(a)}{1-e^{-J(a)}}  +\psi(x)\right] e^{-NJ(a)}.   
\end{split}
\end{equation}
 There are terms of order $1$, $N e^{-NJ(a)}$ and $e^{-NJ(a)}$ in $T(x+Na)$. 
Let us calculate the term of order $1$ in $T(x+Na)$, using the relation $\varphi(a) = C e^{-J(a)}$, the expression of $e(0)$ in the case of non-zero survival probability, and the expression of $t(0)$: 
\begin{equation}\label{simplo}
\begin{split}
t(0)&\left[  1 + \frac{2\Xi_-(a)}{1- e^{-J(a)}}   \right] +
2C\Xi_-(a)\frac{e^{-J(a)}}{(1-e^{-J(a)})^2} + 2\frac{K(a)}{1-e^{-J(a)}} + \frac{G(a)}{1-e^{-J(a)}}\\
=& t(0)\left[  1 + \frac{2\Xi_-(a)}{1- e^{-J(a)}}   \right] +
\frac{2\Xi_-(a)\varphi(a)}{(1-e^{-J(a)})^2} + \frac{\psi(a)}{1-e^{-J(a)}}\\
=&-\frac{t(0)}{e(0)} + \frac{2\Xi_-(a)\varphi(a)}{(1-e^{-J(a)})^2} + \frac{\psi(a)}{1-e^{-J(a)}}\\
=& \frac{-\psi(a) +\left( \frac{1}{e(0)}+1\right)\psi(a)}{ 1 - e^{-J(a)}}
+
\frac{2\Xi_-(a)\varphi(a)}{(1-e^{-J(a)})^2} + \frac{\psi(a)}{1-e^{-J(a)}}\\
=&  \frac{1}{1-e^{-J(a)}}\left( \frac{1}{e(0)} + 1 +    \frac{2\Xi_-(a)}{1-e^{-J(a)}}  \right)   \varphi(a)\\
=&0.
\end{split}
\end{equation}
The leading terms in $t(x+Na)$ and $T(x+Na)$ are therefore both of order $N e^{-NJ(a)}$.
\begin{equation}
\begin{split}
t(x+Na ) =& t(0) e^{-J(Na + x)} + \varphi(x+Na)\\
=& \varphi(a) e^{J(a)}e^{-J(x)} Ne^{-NJ(a)} + ( t(0) + D(x) )e^{-J(x)} e^{-NJ(a)}\\
T(x+Na) =& t(0)( 1 + 2\Xi_-(Na + x ) ) +\psi(x+Na)\\
=&  2\varphi(a) e^{J(a)}\left( \Xi_-(x) - \frac{\Xi_-(a)}{1-e^{-J(a)}} \right)N e^{-NJ(a)}\\
 &+ \left[ 2t(0)\left( -\frac{\Xi_-(a)}{1-e^{-J(a)}} +\Xi_-(x)\right) - 2\frac{\varphi(a)\Xi_-(a)}{(1-e^{-J(a)})^2}- \frac{\psi(a)}{1-e^{-J(a)}}  +\psi(x)\right] e^{-NJ(a)} . 
\end{split}
\end{equation}
To simplify the term of order $e^{-NJ(a)}$ in the expression of $T(x+Na)$, we use the expression of $t(0)$ given in Eq. (\ref{t0Expr}):
\begin{equation}
\begin{split}
 -2t(0)\frac{\Xi_-(a)}{1-e^{-J(a)}}& - 2\frac{\varphi(a)\Xi_-(a)}{(1-e^{-J(a)})^2}- \frac{\psi(a)}{1-e^{-J(a)}}  
 = -2t(0)\frac{\Xi_-(a)}{1-e^{-J(a)}} +\frac{t(0)\left( 1-e^{-J(a)} + 2\Xi_-(a) \right)}{1-e^{-J(a)}}\\
 =& \frac{t(0)}{1-e^{-J(a)}}(-2\Xi_-(a) +  1-e^{-J(a)} + 2\Xi_-(a)) = t(0).
 \end{split}
\end{equation}
Moreover, from the expression of $\varphi(x+Na)$ in  Eqs (\ref{phiExprNx},\ref{CDDef}), we note that $D(x)e^{-J(x)}$ is $\varphi(x)$, for $x$ in $[0,a[$. We obtain
\begin{equation}
\begin{split}
t(x+Na ) 
=& \varphi(a) e^{J(a)}e^{-J(x)} Ne^{-NJ(a)} + [ t(0) e^{-J(x)} + \varphi(x)] e^{-NJ(a)},\\
T(x+Na)
=&  2\varphi(a) e^{J(a)}\left( \Xi_-(x) - \frac{\Xi_-(a)}{1-e^{-J(a)}} \right)N e^{-NJ(a)}\\
 &+ \left[ t(0) ( 1 + 2\Xi_-(x) ) +\psi(x) \right] e^{-NJ(a)}. 
\end{split}
\end{equation}
 As a consistency check, one notices that if  $N=0$, the above two equations reduce to the expression of $t(x)$ and $T(x)$ given in Eq. (\ref{tTExpr}).\\
 Moreover, one can check continuity at $a$ of the two functions $t$ and $T$ by calculating the values
 \begin{equation}
 \begin{split}
\underset{x\to a}{\lim} t( x ) = t(0) e^{-J(a)} + \varphi(a) = t( 0 + a),\\
\underset{x\to a}{\lim} T( x ) = t(0)( 1 + 2\Xi_-(a)) +\psi(a),\\
T( a ) = -\frac{2\varphi(a) \Xi_-(a)}{1 - e^{-J(a)}} + t(0) e^{-J(a)}.\\
\end{split}
\end{equation}
 The last two expressions are easily found to be equal to each other thanks to the expression of $t(0)$ in terms of $\varphi(a)$ and $\psi(a)$ given in Eq. (\ref{t0Expr}).\\ 

The sum and difference of $T(x+Na)$ and $t(x+Na)$ yield
\begin{equation}\label{appTpmp}
\begin{split}
T(x+Na,\pm )
= & \varphi(a) e^{J(a)}\left( 2\Xi_-(x)-2 \frac{\Xi_-(a)}{1-e^{-J(a)}} \pm e^{-J(x) }\right)Ne^{-NJ(a)}\\
&+ \left[ t(0)( 1 \pm e^{-J(x) }+ 2\Xi_-(x) ) +\psi(x) \pm \varphi(x)\right] e^{-NJ(a)},\\
&\qquad(N\in \mathbbm{N}, x \in [0,a[). 
\end{split}
\end{equation}


 \noindent{Substituting the expression of $t(0)$ obtained in Eq. (\ref{t0Expr}) yields Eq. (\ref{Tpmp}) of the main text.}

\section{Consistency checks in the case of a constant drift}\label{checkConst}
In the case of a constant drift, the quantity $\alpha$ is $\frac{2\mu}{1-\mu^2}<0$ for $x$ in $[0,a]$. Evaluating integrals yields
\begin{equation}
J( x ) = \alpha x.
\end{equation}
\begin{equation}
 \Xi_-(x) = \frac{1}{2\mu}( 1- e^{-\alpha x} ).
\end{equation}

\subsection{Almost-sure exit: constant drift $\mu<0$}\label{constantCheckNeg} 
 In the expression of $T(x+Na,\pm)$ in Eq. (\ref{resT1}), let us calculate the coefficient of $N$:
\begin{equation}
\begin{split}
\frac{J(a)}{2} +2\frac{ \Xi_+( a) \Xi_-(a)}{1 - e^{J(a)}} &-2 \int_0^a dv\Xi'_-(v)\Xi_+(v) \\
=&
\frac{\alpha a}{2}+\frac{2}{(2\mu)^2)}\frac{(1-e^{-\alpha a})(e^{-\alpha a}-1)}{1-e^{\alpha a}} -\frac{2}{(1-\mu^2)^2}\int_0^a dv e^{-\alpha v}\frac{1}{\alpha}( e^{\alpha v} - 1) \\
=& a\left( \frac{\alpha}{2} -\frac{1}{\mu(1-\mu^2)} \right)
 + (e^{-\alpha a} - 1) \left[ \frac{1}{2\mu^2} - \frac{1}{\mu(1-\mu^2) \alpha)}\right]\\
 =& a\frac{\mu^2 - 1}{\mu(1 - \mu^2)}\\
 =& -\frac{a}{\mu} = \frac{a}{|\mu|}.
\end{split}
\end{equation}

The term of order $1$ is evaluated as
\begin{equation}
\begin{split}
\frac{\Xi_+(a)}{1-e^{J(a)}}&\left( 1 \pm e^{-J(x)} + 2\Xi_-(x)\right)
 +\frac{1}{2}J(x) \mp e^{-J(x)}\Xi_+(x) - 2 \int_0^x dv\Xi_-'(v)\Xi_+(v)\\
=& \frac{1}{2\mu}\left[ 1 \pm e^{\alpha x}+ \frac{1}{\mu} (1 - e^{-\alpha x})\right] + \frac{\mu}{1-\mu^2} x \mp e^{-\alpha x}(e^{\alpha x}-1) -2\int_0^x dv e^{-\alpha v}\frac{(e^{\alpha v} - 1)}{2\mu} \\
 =& x \frac{1}{1-\mu^2}\left(\mu- \frac{1}{\mu}  \right)   \frac{\mu}{1-\mu^2}x - \frac{1}{2\mu}\left( 1\pm e^{-\alpha x} +\frac{1}{\mu}(1-e^{-\alpha x}) \right) \mp \frac{1}{2\mu}(1-e^{-\alpha x}) + \frac{1}{2\mu^2}(1-e^{-\alpha x})\\
 =& -\frac{x}{\mu} + e^{-\alpha x}\left( \mp \frac{1}{2\mu} 
 + \frac{1}{2\mu^2} 
 \pm \frac{1}{2\mu} -  \frac{1}{2\mu^2} \right)
 -  \frac{1}{2\mu} - \frac{1}{2\mu^2} \mp \frac{1}{2\mu} + \frac{1}{2\mu^2}  \\
 =& \frac{x}{|\mu|} - \frac{1\pm 1}{2\mu},
 \qquad\qquad(x\in [0,a[).
\end{split}
\end{equation} 
Combining the two terms yields
\begin{equation}
\begin{split}
\langle T( x+Na) \rangle_{c,+} =T(Na + x,+) =& \frac{Na + x}{|\mu|} + \frac{1}{|\mu|},\\
\langle T( x+Na) \rangle_{c,-} = T(Na + x,-) =& \frac{Na + x}{|\mu|},\qquad(x\in [0,a[, N\in \mathbb{N}),
\end{split}
\end{equation} 
 which is the known result in the case of a constant negative drift.\\

\subsection{Non-zero survival probability: constant drift $\mu>0$} 
Evaluating integrals in the expression of the exit probability in Eq. (\ref{EqNapm}) yields for $x\geq 0$:
\begin{equation}\label{explExitConst}
\begin{split}
 \frac{2\Xi_-(a)}{1 - e^{-J(a)}} =& \frac{1}{\mu},\;\;\;\;\;\;e(0) = -\frac{\mu}{1+\mu},~\;\;\;\;\;\;
 E(0)=\frac{1}{1+\mu},\\
 e(x) =& -\frac{\mu}{1+\mu} e^{-\frac{2\mu x}{1-\mu^2}},\\
 E(x) =& \frac{1}{1+\mu} - 2\frac{\mu}{1+\mu}\Xi_-(x) = \frac{1}{1+\mu}e^{-\frac{2\mu x}{1-\mu^2}}.
 \end{split}
\end{equation}
Hence
\begin{equation}\label{repro}
\begin{split}
E(x,+) =&  \frac{1-\mu}{1+\mu} e^{-\frac{2\mu x}{1-\mu^2}},\\
E(x,-) =& e^{-\frac{2\mu x}{1-\mu^2}},\qquad(x\geq 0).  
\end{split}
\end{equation}
 This reproduces the result $(\mathrm{I})_{\mu}$ {\textcolor{black}{(Eq. (\ref{EConst}))}} obtained in \cite{de2021survival} (Eqs (45a) and (45b) there) for a constant drift.\\

 To check the value of $\langle T(x) \rangle_{c,\pm}$, we evaluate integrals with the above expressions, which yields
\begin{equation}
 \varphi(x) = - \frac{1+\mu^2}{(1+\mu)(1-\mu^2)} x e^{-\alpha x},\quad x\in[0,a].
\end{equation} 
\begin{equation}
\begin{split}
\psi(x) =& - 2\frac{1+\mu^2}{(1+\mu)(1-\mu^2)^2} 
\int_0^x dy \,y e^{-\alpha y} + \frac{2\mu}{(1+\mu)(1-\mu^2)^2} \int_0^x dy e^{-\alpha y}\\
=& \frac{1+\mu^2}{\mu(1+\mu)(1-\mu^2)}x e^{-\alpha x} + \frac{\mu^2 - 1}{2\mu^2(1+\mu)}(1-e^{-\alpha x}).\\
\end{split}
\end{equation}
Moreover, $t(0) = \frac{1}{2\mu}\frac{1-\mu}{1+\mu}$, hence the 
 coefficient of $Ne^{-NJ(a)}$
 in the expressions obtained in Eq. (\ref{Tpmp}) to evaluate $T(x+Na,+)$ and $T(x+Na,-)$:\\  
\begin{equation}\label{domTerm}
\begin{split}
  \varphi(a) e^{J(a)}\left( 2\Xi_-(x)-2 \frac{\Xi_-(a)}{1-e^{-J(a)}} - e^{-J(x) }\right)=& 
  -\frac{1+\mu^2}{(1+\mu)(1-\mu^2)}a\left( \frac{1}{\mu}(1-e^{-\alpha x})-\frac{1}{\mu} - e^{-\alpha x} \right)\\
  =& \frac{1+\mu^2}{\mu(1-\mu^2)} a e^{-\alpha x},\\
   \varphi(a) e^{J(a)}\left( 2\Xi_-(x)-2 \frac{\Xi_-(a)}{1-e^{-J(a)}} + e^{-J(x) }\right)=& 
  -\frac{1+\mu^2}{(1+\mu)(1-\mu^2)}a\left( \frac{1}{\mu}(1-e^{-\alpha x})-\frac{1}{\mu} + e^{-\alpha x} \right)\\
  =& \frac{1+\mu^2}{\mu(1+\mu)^2} a e^{-\alpha x}.
\end{split}
\end{equation}
 Let us evaluate the quantities $T(x)$ and $t(x)$ for $x\in [0,a[$:
\begin{equation}
\begin{split}
 t(x) =&  t(0) e^{-J(x)} + \varphi(x)\\
 =& \frac{1}{2\mu}\frac{1-\mu}{1+\mu} e^{-\alpha x}
 - \frac{1+\mu^2}{(1+\mu)(1-\mu^2)} x e^{-\alpha x}, \quad x\in [0,a[,\\
\end{split}
\end{equation}
\begin{equation}
\begin{split}
 T(x) =&  t(0) + 2t(0) \Xi_-(x) +\psi(x)\\
 =& \frac{1}{2\mu}\frac{1-\mu}{1+\mu} +
  \frac{1}{2\mu^2}\frac{1-\mu}{1+\mu}( 1-e^{-\alpha x}) + 
   \frac{1+\mu^2}{\mu(1+\mu)(1-\mu^2)}x e^{-\alpha x} + \frac{\mu^2 - 1}{2\mu^2(1+\mu)}(1-e^{-\alpha x})\\
=& \frac{1+\mu^2}{\mu(1+\mu)(1-\mu^2)}x e^{-\alpha x} + \frac{1-\mu}{2\mu( 1 + \mu)}e^{-\alpha x}, \quad x\in [0,a[.
\end{split}
\end{equation}
Hence
\begin{equation}
\begin{split}
T( x, -) =& \frac{1+\mu^2}{(1+\mu)(1-\mu^2)}\left( \frac{1}{\mu} + 1 
  \right)x e^{-\alpha x}\\
  =&  \frac{1+\mu^2}{\mu(1-\mu^2)} x e^{-\alpha x},\\
T(x, + ) =& \frac{1+\mu^2}{(1+\mu)(1-\mu^2)}\left( \frac{1}{\mu} - 1 
  \right)x e^{-\alpha x} +\frac{1-\mu}{\mu(1 + \mu)} e^{-\alpha x}  \\
  =&   \frac{1+\mu^2}{\mu(1+\mu)^2} x e^{-\alpha x}
    +\frac{1-\mu}{\mu(1 + \mu)} e^{-\alpha x},
   \quad x\in [0,a[.
\end{split}    
\end{equation}
Combining with Eq (\ref{domTerm}) yields
\begin{equation}
\begin{split}
T( x+Na, -) =& \frac{1+\mu^2}{\mu( 1 - \mu^2)}
(  Na + x ) e^{-\alpha(Na+ x)},\\
T( x+Na, + ) =& \frac{1+\mu^2}{\mu(1+\mu)^2}(Na + x) x e^{-\alpha( Na + x) x} +\frac{1-\mu}{\mu(1 + \mu)} e^{-\alpha(Na + x)},  \\
  &  \quad ( N \in \mathbbm{N}, x\in [0,a[).
\end{split}    
\end{equation}
 Dividing these two expressions by $E(x,-) = e^{-\alpha( Na +x)}$ and $E(x,-) = \frac{1-\mu}{1+\mu}e^{-\alpha( Na +x)}$ respectively yields
 \begin{equation}
\begin{split}
\langle T( x+Na) \rangle_{c,-}=& \frac{1+\mu^2}{\mu( 1 - \mu^2)}
(  Na + x ),\\
\langle T( x+Na)\rangle_{c,+} =& \frac{1+\mu^2}{\mu(1-\mu)^2}(Na + x)  +\frac{1}{\mu},  \quad ( N \in \mathbbm{N}, x\in [0,a[),
\end{split}    
\end{equation}
 which is the known result in the case of a constant positive drift (Eq. (\ref{resIImu})).\\

\end{appendices}

\bibliography{bibRef} 
\bibliographystyle{ieeetr}

\end{document}